\documentclass[a4paper,USenglish,cleveref, autoref]{lipics-v2019}

\title{Solving connectivity problems parameterized by treedepth in single-exponential time and polynomial space}
\titlerunning{Solving connectivity problems parameterized by treedepth}

\author{Falko Hegerfeld}{Humboldt-Universit\"at zu Berlin, Germany}{hegerfeld@informatik.hu-berlin.de}{https://orcid.org/0000-0003-2125-5048}{}
\author{Stefan Kratsch}{Humboldt-Universit\"at zu Berlin, Germany}{kratsch@informatik.hu-berlin.de}{https://orcid.org/0000-0002-0193-7239}{}

\authorrunning{F. Hegerfeld and S. Kratsch}

\Copyright{Falko Hegerfeld and Stefan Kratsch}



\nolinenumbers 

\hideLIPIcs  

\EventEditors{John Q. Open and Joan R. Access}
\EventNoEds{2}
\EventLongTitle{42nd Conference on Very Important Topics (CVIT 2016)}
\EventShortTitle{CVIT 2016}
\EventAcronym{CVIT}
\EventYear{2016}
\EventDate{December 24--27, 2016}
\EventLocation{Little Whinging, United Kingdom}
\EventLogo{}
\SeriesVolume{42}
\ArticleNo{23}

\usepackage{xspace}
\usepackage[linesnumbered,ruled,vlined]{algorithm2e}

\SetCommentSty{mycommstyle}
\SetKw{True}{true}
\SetKw{False}{false}

\theoremstyle{definition}
\newtheorem{dfn}{Definition}[section]

\theoremstyle{plain}
\newtheorem{thm}[dfn]{Theorem}
\newtheorem{cor}[dfn]{Corollary}
\newtheorem{lem}[dfn]{Lemma}

\crefname{dfn}{Definition}{Definitions}
\crefname{thm}{Theorem}{Theorems}
\crefname{cor}{Corollary}{Corollaries}
\crefname{lem}{Lemma}{Lemmata}
\crefname{prop}{Proposition}{Propositions}
\crefname{rem}{Remark}{Remarks}
\crefname{algorithm}{Algorithm}{Algorithms}
\crefname{algocf}{Algorithm}{Algorithms}

\crefformat{equation}{(#2#1#3)}

\makeatletter
\newenvironment{problem}[2][]{%
  \def\problem@arg{#1}%
  \def\problem@framed{framed}%
  \def\problem@lined{lined}%
  \def\problem@doublelined{doublelined}%
  \ifx\problem@arg\@empty%
    \def\problem@hline{}%
  \else%
    \ifx\problem@arg\problem@doublelined%
      \def\problem@hline{\hline\hline}%
    \else%
      \def\problem@hline{\hline}%
    \fi%
  \fi%
  \ifx\problem@arg\problem@framed%
    \def\problem@table{\tabularx{\textwidth}{|>{\bfseries}lX|c}}%
    \def\problem@title{\multicolumn{2}{|l|}{%
        \raisebox{-\fboxsep}{\textsc{\Large #2}}%
      }}%
  \else
    \def\problem@table{\tabularx{\textwidth}{>{\bfseries}lXc}}%
    \def\problem@title{\multicolumn{2}{l}{%
        \raisebox{-\fboxsep}{\textsc{\Large #2}}%
      }}%
  \fi%
  \bigskip\par\noindent%
    \problem@table%
      \problem@hline%
      \problem@title\\[2\fboxsep]%
}{%
      \\\problem@hline%
    \endtabularx%
  \medskip\par%
}
\makeatother

\newcommand{\Oh}{\mathcal{O}} 

\newcommand{\NP}{\ensuremath{\textsf{NP}}\xspace}
\newcommand{\FPT}{\ensuremath{\textsf{FPT}}\xspace}

\newcommand{\family}{\mathcal{F}}
\newcommand{\solutions}{\mathcal{S}}
\newcommand{\candidates}{\mathcal{R}}

\newcommand{\cuts}{\mathcal{C}}
\newcommand{\cutsols}{\mathcal{Q}}
\newcommand{\tree}{\mathtt{tree}}
\newcommand{\tail}{\mathtt{tail}}
\newcommand{\child}{\mathtt{child}}
\newcommand{\broom}{\mathtt{broom}}
\newcommand{\ZZ}{\mathbb{Z}}
\newcommand{\PP}{\mathbb{P}}
\newcommand{\NN}{\mathbb{N}}
\newcommand{\FF}{\mathbb{F}}
\newcommand{\TT}{\mathcal{T}}
\newcommand{\one}{\mathbf{1}}
\newcommand{\zero}{\mathbf{0}}
\newcommand{\wfct}{\mathbf{w}}
\newcommand{\wtarget}{w}
\newcommand{\countc}{\mathtt{CountC}}

\newcommand{\cc}{\mathtt{cc}}

\newcommand{\Pin}[2]{P_{[#1]}(#2)}
\newcommand{\Pex}[2]{P_{(#1)}(#2)}
\newcommand{\partialin}[2]{\mathcal{P}_{[#1]}(#2)}
\newcommand{\partialex}[2]{\mathcal{P}_{(#1)}(#2)}

\newcommand{\td}{\mathrm{td}}
\newcommand{\tw}{\mathrm{tw}}
\newcommand{\pw}{\mathrm{pw}}

\newcommand{\bip}{\mathbf{bip}}

\newcommand{\forest}{\mathbf{F}}
\newcommand{\marked}{\mathbf{M}}
\newcommand{\deleted}{\mathbf{X}}
\newcommand{\remain}{\mathbf{A}}

\newcommand{\states}{\mathtt{states}}

\newcommand{\sep}{\mid}
\newcommand{\cupdot}{\mathbin{\mathaccent\cdot\cup}}

\begin{document}

\maketitle

\begin{abstract}
A breakthrough result of Cygan et al.\ (FOCS 2011) showed that connectivity problems parameterized by treewidth can be solved much faster than the previously best known time $\Oh^*(2^{\Oh(\tw\log \tw)})$. Using their inspired Cut\&Count technique, they obtained $\Oh^*(\alpha^\tw)$ time algorithms for many such problems. Moreover, they proved these running times to be optimal assuming the Strong Exponential-Time Hypothesis. Unfortunately, like other dynamic programming algorithms on tree decompositions, these algorithms also require \emph{exponential space}, and this is widely believed to be unavoidable. In contrast, for the slightly larger parameter called treedepth, there are already several examples of algorithms matching the time bounds obtained for treewidth, but using only polynomial space. Nevertheless, this has remained open for connectivity problems.

In the present work, we close this knowledge gap by applying the Cut\&Count technique to graphs of small treedepth. While the general idea is unchanged, we have to design novel procedures for counting consistently cut solution candidates using only polynomial space. Concretely, we obtain time $\Oh^*(3^d)$ and polynomial space for \textsc{Connected Vertex Cover}, \textsc{Feedback Vertex Set}, and \textsc{Steiner Tree} on graphs of treedepth $d$. Similarly, we obtain time $\Oh^*(4^d)$ and polynomial space for \textsc{Connected Dominating Set} and \textsc{Connected Odd Cycle Transversal}.
\end{abstract}

\newpage

\section{Introduction}

The goal of parameterized complexity is to reign in the combinatorial explosion present in \NP-hard problems with the help of a secondary parameter. This leads us to the search for fixed-parameter tractable (\FPT) algorithms, i.e., algorithms with running time $\Oh(f(k) n^c)$ where $n$ is the input size, $k$ is the secondary parameter, $f$ is a computable function, and $c$ is a constant. There are several books giving a broad overview of parameterized complexity \cite{CyganFKLMPPS15, DowneyF13, FlumG06, Niedermeier06}. One of the success stories of parameterized complexity is a graph parameter called treewidth. A large swath of graph problems admit \FPT-algorithms when parameterized by treewidth as witnessed by, amongst other things, Courcelle's theorem \cite{Courcelle90}. However, the function $f$ resulting from Courcelle's theorem is non-elementary \cite{FrickG04}. Thus, a natural goal is to find algorithms with a smaller, or ideally minimal, dependence on the treewidth in the running time, i.e., algorithms where $f$ is as small as possible. Problems only involving local constraints usually permit a single-exponential dependence on the treewidth ($\tw$) in the running time, i.e., time $\Oh^*(\alpha^{\tw})$ for some small constant $\alpha$,\footnote{The $\Oh^*$-notation hides polynomial factors in the input size.} by means of dynamic programming on tree decompositions~\cite{AlberN02, TelleP93, RooijBR09, RooijBLRV18}. For many of these problems we also know the optimal base $\alpha$ if we assume the Strong Exponential-Time Hypothesis (SETH) \cite{LokshtanovMS18}. For a long time a single-exponential running time seemed to be out of reach for problems involving global constraints, in particular for connectivity constraints. This changed when Cygan et al.~\cite{CyganNPPRW11} introduced the Cut\&Count technique, which allowed them to obtain single-exponential-time algorithms for many graph problems involving connectivity constraints. Again, many of the resulting running times can be shown to be optimal assuming SETH \cite{CyganNPPRW11}. 

The issue with treewidth-based algorithms is that dynamic programming on tree decompositions seems to inherently require exponential space. In particular, Chen et al.~\cite{ChenRRV18} devised a model for single-pass dynamic programming algorithms on tree decompositions and showed that such algorithms require exponential space for \textsc{Vertex Cover} and 3-\textsc{Coloring}. Algorithms requiring exponential time and exponential space usually run out of available space before they hit their time limit \cite{Woeginger04}. Hence, it is desirable to reduce the space requirement while maintaining the running time. As discussed, this seems implausible for treewidth. Instead, we consider a different, but related, parameter called treedepth. Treedepth is a slightly larger parameter than treewidth and of great importance in the theory of sparse graphs \cite{NesetrilM06, NesetrilO12, NesetrilM15}. It has been studied under several names such as minimum elimination tree  height \cite{BodlaenderGHK95}, ordered chromatic number \cite{KatchalskiMS95}, and vertex ranking \cite{BodlaenderDJKKMT98}. F\"urer and Yu~\cite{FurerY17} established an explicit link between treedepth and tree decompositions, namely that treedepth is obtained by minimizing the maximum number of forget nodes in a root-leaf-path over all nice tree decompositions (see \cite{Kloks94} for a definition). Many problems parameterized by treedepth allow branching algorithms on \emph{elimination forests}, also called \emph{treedepth decompositions}, that match the running time of the treewidth-algorithms, but replacing the dependence on treewidth by treedepth, while only requiring polynomial space \cite{ChenRRV18, FurerY17, PilipczukW18}. 

\subparagraph*{Our contribution.}
The Cut\&Count technique reduces problems with connectivity constraints to counting problems of certain cuts, called \emph{consistent cuts}. We show that for several connectivity problems the associated problem implied by the Cut\&Count technique can be solved in time $\Oh^*(\alpha^d)$ and polynomial space, where $\alpha$ is a constant and $d$ is the depth of a given elimination forest. Furthermore, the base $\alpha$ matches the base in the running time of the corresponding treewidth-algorithm. Concretely, given an elimination forest of depth $d$ for a graph $G$ we prove the following results:
\begin{itemize}
 \item \textsc{Connected Vertex Cover}, \textsc{Feedback Vertex Set}, and \textsc{Steiner Tree} can be solved in time $\Oh^*(3^d)$ and polynomial space.
 \item \textsc{Connected Dominating Set} and \textsc{Connected Odd Cycle Transversal} can be solved in time $\Oh^*(4^d)$ and polynomial space.
\end{itemize}

\subparagraph*{Related work.}
The Cut\&Count technique leads to randomized algorithms as it relies on the Isolation Lemma. At the cost of a worse base in the running time, Bodlaender et al.~\cite{BodlaenderCKN15} present a generic method, called the rank-based approach, to obtain deterministic single-exponential-time algorithms for connectivity problems parameterized by treewidth; the rank-based approach is also able to solve counting variants of several connectivity problems. Fomin et al.~\cite{FominLPS16} use matroid tools to, amongst other results, reobtain the deterministic running times of the rank-based approach. In a follow-up paper, Fomin et al.~\cite{FominLPS17} manage to improve several of the deterministic running times using their matroid tools. Multiple papers adapt the Cut\&Count technique and rank-based approach to graph parameters different from treewidth. Bergougnoux and Kant\'e~\cite{BergougnouxK19a} apply the rank-based approach to obtain single-exponential-time algorithms for connectivity problems parameterized by cliquewidth. The same authors~\cite{BergougnouxK19b} generalize, incurring a loss in the running time, this approach to a wider range of parameters including rankwidth and mim-width. Furthermore, Bergougnoux~\cite{Bergougnoux19} also applies the Cut\&Count technique in this more general setting. Pino et al.~\cite{PinoBR16} use the Cut\&Count technique and rank-based approach to obtain fast deterministic and randomized algorithms for connectivity problems parameterized by branchwidth. 

Lokshtanov and Nederlof~\cite{LokshtanovN10} present a framework using algebraic techniques, such as Fourier, Möbius, and Zeta transforms, to reduce the space usage of certain dynamic programming algorithms from exponential to polynomial. F\"urer and Yu~\cite{FurerY17} adapt this framework to the setting where the underlying set (or graph) is dynamic instead of static, in particular for performing dynamic programming along the bags of a tree decomposition, and obtain a $\Oh^*(2^d)$-time, where $d$ is the depth of a given elimination forest, and polynomial-space algorithm for counting perfect matchings. Using the same approach, Belbasi and F\"urer~\cite{BelbasiF19} design an algorithm counting the number of Hamiltonian cycles in time $\Oh^*((4k)^d)$, where $k$ is the width and $d$ the depth of a given tree decomposition, and polynomial space. Furthermore, they also present an algorithm for the traveling salesman problem with the same running time, but requiring pseudopolynomial space.

\subparagraph*{Follow-up work.}\label{followup}
Nederlof et al.~\cite{NederlofPSW20} solve the open problem stated in the conclusion regarding a polynomial-space algorithm for $\textsc{Hamiltonian Cycle}$ and $\textsc{Hamiltonian Path}$ parameterized by treedepth. They obtain essentially a single algorithm running in time $\Oh^*(5^d)$ and polynomial space for $\textsc{Hamiltonian Cycle}$, $\textsc{Hamiltonian Path}$, $\textsc{Long Cycle}$, $\textsc{Long Path}$ and $\textsc{Minimum Cycle Cover}$. In this case the obtained base is worse than the base for these problems parameterized by treewidth; they can all be solved in time $\Oh^*(4^\tw)$ \cite{CyganNPPRW11} and $\textsc{Hamiltonian Cycle}$ and $\textsc{Hamiltonian Path}$ can even be solved in time $\Oh^*((2 + \sqrt{2})^\pw)$ parameterized by pathwidth \cite{CyganKN18}.

\subparagraph*{Organization.}
We describe the preliminary definitions and notations in \cref{sec:prelims}. In \cref{sec:cutandcount} we first discuss the Cut\&Count setup and give a detailed exposition for \textsc{Connected Vertex Cover}. Afterwards, we explain what general changes can occur for the other problems and then discuss the remaining problems \textsc{Feedback Vertex Set}, \textsc{Connected Dominating Set}, \textsc{Steiner Tree}, and \textsc{Connected Odd Cycle Transversal}. We conclude in \cref{sec:conclusion}.

\section{Preliminaries}
\label{sec:prelims}

\subsection{Notation}

Let $G = (V, E)$ be an undirected graph. We denote the number of vertices by $n$ and the number of edges by $m$. For a vertex set $X \subseteq V$, we denote by $G[X]$ the subgraph of $G$ that is induced by $X$. The \emph{open neighborhood} of a vertex $v$ is given by $N(v) = \{u \in V \sep \{u,v\} \in E\}$, whereas the \emph{closed neighborhood} is given by $N[v] = N(v) \cup \{v\}$. We extend these notations to sets $X \subseteq V$ by setting $N[X] = \bigcup_{v \in X} N[v]$ and $N(X) = N[X] \setminus X$. Furthermore, we denote by $\cc(G)$ the number of connected components of $G$. For two disjoint sets $A,B \subseteq V$, we let $E(A,B) = \{\{a,b\} \in E \sep a \in A, b \in B\}$ denote the set of edges with one endpoint in $A$ and one endpoint in $B$.

A \emph{cut} of a set $X \subseteq V$ is a pair $(X_L, X_R)$ with $X_L \cap X_R = \emptyset$ and $X_L \cup X_R = X$, we also use the notation $X = X_L \cupdot X_R$. We refer to $X_L$ and $X_R$ as the \emph{left} and \emph{right side} of the cut, respectively. Note that either side may be empty, although usually the left side is nonempty. A \emph{cut} of a graph $G=(V,E)$ is a cut of the vertex set $V$.

For two integers $a,b$ we write $a \equiv b$ to indicate equality modulo 2, i.e., $a$ is even if and only if $b$ is even. We use Iverson's bracket notation: for a predicate $p$, we have that $[p]$ is $1$ if $p$ is true and $0$ otherwise. Given a function $f$, we denote by $f[v \mapsto \alpha]$ the function $(f \setminus \{(v, f(v))\}) \cup \{(v, \alpha)\}$. By $\FF_2$ we denote the field of two elements. For a field or ring $\FF$ we denote by $\FF[Z_1, Z_2, \ldots, Z_t]$ the ring of polynomials in the indeterminates $Z_1, Z_2, \ldots, Z_t$ with coefficients in $\FF$. The $\Oh^*$-notation hides polynomial factors in the input size. For a natural number $n$, we denote by $[n]$ the set of integers from $1$ to $n$.

\subsection{Treedepth}
\label{sec:treedepth}

\begin{dfn}
 An \emph{elimination forest} of an undirected graph $G = (V, E)$ is a rooted forest $\TT = (V, E_\TT)$ such that for every edge $\{u, v\} \in E$ either $u$ is an ancestor of $v$ in $\TT$ or $v$ is an ancestor of $u$ in $\TT$. The \emph{depth} of a rooted forest is the largest number of nodes on a path from a root to a leaf. The \emph{treedepth} of $G$ is the minimum depth over all elimination forests of $G$ and is denoted by $\td(G)$. 
\end{dfn}
 
 We slightly extend the notation for elimination forests used by Pilipczuk and Wrochna~\cite{PilipczukW18}. For a rooted forest $\TT = (V, E_\TT)$ and a node $v \in V$ we denote by $\tree[v]$ the set of nodes in the subtree rooted at $v$, including $v$. By $\tail[v]$ we denote the set of all ancestors of $v$, including $v$. Furthermore, we define $\tree(v) = \tree[v] \setminus \{v\}$, $\tail(v) = \tail[v] \setminus \{v\}$, and $\broom[v] = \{v\} \cup \tail(v) \cup \tree(v)$. By $\child(v)$ we denote the children of $v$.

Note that an elimination forest $\TT$ of a connected graph consists only of a single tree.

\subsection{Isolation Lemma}

\begin{dfn}
 A function $\wfct \colon U \rightarrow \ZZ$ \emph{isolates} a set family $\family \subseteq 2^U$ if there is a unique $S' \in \family$ with $\wfct(S') = \min_{S \in \family} \wfct(S)$, where for subsets $X \subseteq U$ we define $\wfct(X) = \sum_{u \in X} \wfct(u)$.
\end{dfn}

\begin{lem}[Isolation Lemma, \cite{MulmuleyVV87}]
 \label{thm:isolation}
 Let $\family \subseteq 2^U$ be a nonempty set family over a universe $U$. Let $N \in \NN$ and for each $u \in U$ choose a weight $\wfct(u) \in [N]$ uniformly and independently at random. Then
  $\PP[\wfct \text{ isolates } \family] \geq 1 - |U|/N$.
\end{lem}

When counting objects modulo 2 the Isolation Lemma allows us to avoid unwanted cancellations by ensuring with high probability that there is a unique solution. In our applications, we will choose $N$ so that we obtain an error probability of less than $1/2$.

\section{Cut\&Count}
\label{sec:cutandcount}

In this section $G = (V, E)$ always refers to a connected undirected graph. For the sake of a self-contained presentation, we state the required results for the Cut\&Count technique again, mostly following the presentation of Cygan et al.~\cite{CyganNPPRW11}. Our approach only differs from that of Cygan et al.~\cite{CyganNPPRW11} in the counting sub-procedure. 

We begin by describing the Cut\&Count setup and then present the counting sub-procedure for \textsc{Connected Vertex Cover}. Afterwards we explain how to adapt the counting sub-procedure for the other problems. Our exposition is the most detailed for \textsc{Connected Vertex Cover}, whereas the analogous parts of the other problems will not be discussed in such detail.

\subsection{Setup}

Suppose that we want to solve a problem on $G$ involving connectivity constraints, then we can make the following general definitions. The solutions to our problem are subsets of a universe $U$ which is related to $G$. Let $\solutions \subseteq 2^U$ denote the set of solutions and we want to determine whether $\solutions$ is empty or not. The Cut\&Count technique consists of two parts:
\begin{itemize}
 \item \textbf{The Cut part:} We relax the connectivity constraints to obtain a set $\solutions \subseteq \candidates \subseteq 2^U$ of possibly connected solutions. The set $\cutsols$ will contain pairs $(X, C)$ consisting of a candidate solution $X \in \candidates$ and a consistent cut $C$  of $X$, which is defined in \cref{dfn:cons_cut}.
 \item \textbf{The Count part:} We compute $|\cutsols|$ modulo 2 using a sub-procedure. The consistent cuts are defined so that disconnected candidate solutions $X \in \candidates \setminus \solutions$ cancel, because they are consistent with an even number of cuts. Hence, only connected candidates $X \in \solutions$ remain.
\end{itemize}
If $|\solutions|$ is even, then this approach does not work, because the connected solutions would cancel out as well when counting modulo 2. To circumvent this difficulty, we employ the Isolation Lemma (\cref{thm:isolation}). By sampling a weight function $\wfct \colon U \rightarrow [N]$, we can instead count pairs with a fixed weight and it is likely that there is a weight with a unique solution if a solution exists at all. Formally, we compute $|\cutsols_\wtarget|$ modulo 2 for every possible weight $\wtarget$, where $\cutsols_\wtarget = \{(X, C) \in \cutsols \sep \wfct(X) = \wtarget\}$, instead of computing $|\cutsols|$ modulo 2.
\begin{dfn}[\cite{CyganNPPRW11}]
 \label{dfn:cons_cut}
 A cut $(V_L, V_R)$ of an undirected graph $G = (V, E)$ is \emph{consistent} if $u \in V_L$ and $v \in V_R$ implies $\{u,v\} \notin E$, i.e., $E(V_L, V_R) = \emptyset$. A \emph{consistently cut subgraph} of $G$ is a pair $(X, (X_L, X_R))$ such that $X \subseteq V$ and $(X_L, X_R)$ is a consistent cut of $G[X]$. For $V' \subseteq V$, we denote the set of consistently cut subgraphs of $G[V']$ by $\cuts(V')$.
\end{dfn}

To ensure that connected solutions are not compatible with an even number of consistent cuts, we will usually force a single vertex to the left side of the consistent cut. This results in the following fundamental property of consistent cuts.
\begin{lem}[\cite{CyganNPPRW11}]
 \label{thm:cons_cut}
 Let $X$ be a subset of vertices such that $v_1 \in X \subseteq V$. The number of consistently cut subgraphs $(X, (X_L, X_R))$ such that $v_1 \in X_L$ is equal to $2^{\cc(G[X]) - 1}$.
\end{lem}

\begin{proof}
 By the definition of a consistently cut subgraph $(X, (X_L, X_R))$, we have for every connected component $C$ of $G[X]$ that either $C \subseteq X_L$ or $C \subseteq X_R$. The connected component $C$ that contains $v_1$ must satisfy $C \subseteq X_L$ and for all other connected components we have 2 choices. Hence, we obtain $2^{\cc(G[X]) - 1}$ different consistently cut subgraphs $(X, (X_L, X_R))$ with $v_1 \in X_L$.
\end{proof}

With \cref{thm:cons_cut} we can distinguish disconnected candidates from connected candidates by determining the parity of the number of consistent cuts for the respective candidate. We determine this number not for a single candidate but we determine the total for all candidates with a fixed weight. \cref{thm:general_setup} encapsulates the Cut\&Count technique for treedepth.

\begin{algorithm}[h]
      \KwIn{Set $U$, elimination forest $\TT$, procedure $\countc$ accepting $\wfct \colon U \rightarrow [N]$, $\wtarget \in \NN$}
      \For{$v \in U$}
      {
	Choose $\wfct(v) \in [2|U|]$ uniformly at random\;
      }
      \For{$\wtarget = 1, \ldots, 2|U|^2$}
      {
	\lIf{$\countc(\wfct, \wtarget, \TT) \equiv 1$}
	{
	  \!\!\Return \True
	}
      }
      \Return \False\;
      \caption{Cut\&Count}
      \label{algo:cutncount}
\end{algorithm}

\begin{cor}
 \label{thm:general_setup}
 Let $\solutions \subseteq 2^U$ and $\cutsols \subseteq 2^{U \times (V \times V)}$ such that the following two properties hold for every weight function $\wfct \colon U \rightarrow [2|U|]$ and target weight $\wtarget \in \NN$:
 \begin{enumerate}
  \item $|\{(X,C) \in \cutsols \sep \wfct(X) = \wtarget\}| \equiv |\{X \in \solutions \sep \wfct(X) = \wtarget\}|$,
  \item There is an algorithm $\countc(\wfct, \wtarget, \TT)$ accepting weights $\wfct \colon U \rightarrow [N]$, a target weight $\wtarget$, and an elimination forest $\TT$, such that $\countc(\wfct, \wtarget, \TT) \equiv |\{(X, C) \in \cutsols \sep \wfct(X) = \wtarget\}|$.
 \end{enumerate}
 Then \cref{algo:cutncount} returns $\mathbf{false}$ if $\solutions$ is empty and $\mathbf{true}$ with probability at least $1/2$ otherwise. 
\end{cor}

\begin{proof}
 By setting $\family = \solutions$ and $N = 2|U|$ in \cref{thm:isolation}, we see that there exists a weight $\wtarget$ such that $|\{X \in \solutions \sep \wfct(X) = \wtarget\}| = 1$ with probability at least $1/2$, unless $\solutions$ is empty. Thus, \cref{algo:cutncount} returns $\mathbf{true}$ with probability at least $1/2$.
 
 If $\solutions$ is empty, then $\countc(\wfct, \wtarget, \TT) \equiv 0$ for all choices of $\wfct$, $\wtarget$, and $\TT$ by property 1. and 2., hence \cref{algo:cutncount} returns $\mathbf{false}$.
\end{proof}

We will use the same definitions as Cygan et al.~\cite{CyganNPPRW11} for $\cutsols$ and $\solutions$, hence it follows from their proofs that Condition 1 in \cref{thm:general_setup} is satisfied. Our contribution is to provide the counting procedure $\countc$ for problems parameterized by treedepth.  

Given the sets $\solutions$, $\candidates$, and $\cutsols$, and a weight function $\wfct \colon U \rightarrow [N]$, we will define for every weight $\wtarget$ the sets $\solutions_\wtarget = \{ X \in \solutions \sep \wfct(X) = \wtarget\}$, $\candidates_\wtarget = \{ X \in \candidates \sep \wfct(X) = \wtarget\}$, and $\cutsols_\wtarget = \{ (X, C) \in \cutsols \sep \wfct(X) = \wtarget\}$.

\subsection{Connected Vertex Cover}
\label{sec:cvc}
\begin{problem}[doublelined]{Connected Vertex Cover}
 Input: & An undirected graph $G = (V, E)$ and an integer $k$.
 \\
 Question: & Is there a set $X \subseteq V$, $|X| = k$, such that $G[X]$ is connected and $X$ is a vertex cover of $G$, i.e., $e \cap X \neq \emptyset$ for all $e \in E$?
\end{problem}

In the considered problems, one usually seeks a solution of size at most $k$. For convenience we choose to look for a solution of size exactly $k$ and solve the other case in the obvious way. We define the objects needed for Cut\&Count in the setting of \textsc{Connected Vertex Cover}. We let $U = V$ and define the candidate solutions by $\candidates = \{ X \subseteq V \sep X \text{ is a vertex cover of } G \text{ and } |X| = k\}$, and the solutions are given by $\solutions = \{ X \in \candidates \sep G[X] \text{ is connected} \}$.

To ensure that a connected solution is consistent with an odd number of cuts, we choose a vertex $v_1$ that is always forced to the left side of the cut (cf.\ \cref{thm:cons_cut}). As we cannot be sure that there is a minimum connected vertex cover containing $v_1$, we take an edge $\{u, v\} \in E$ and run \cref{algo:cutncount} once for $v_1 := u$ and once for $v_1 := v$. Hence, for a fixed choice of $v_1$ we define the candidate-cut-pairs by $\cutsols = \{(X, (X_L, X_R)) \in \cuts(V) \sep X \in \candidates \text{ and } v_1 \in X_L\}$. We must check that these definitions satisfy the requirements of \cref{thm:general_setup}. 
  
\begin{lem}[\cite{CyganNPPRW11}]
 \label{thm:cvc_correct}
 Let $\wfct \colon V \rightarrow [N]$ be a weight function, and let $\cutsols$ and $\solutions$ be as defined above. Then we have for every $\wtarget \in \NN$ that $|\solutions_\wtarget| \equiv |\cutsols_\wtarget|$.
\end{lem}

\begin{proof}
 \cref{thm:cons_cut} implies that $|\cutsols_\wtarget| = \sum_{X \in \candidates_\wtarget} 2^{\cc(G[X]) - 1}$. Hence, $|\cutsols_\wtarget| \equiv |\{ X \in \candidates_\wtarget \sep \cc(G[X]) = 1 \}| = |\solutions_\wtarget|$.
\end{proof}

Next, we describe the procedure $\countc$ for \textsc{Connected Vertex Cover}.

\begin{center}
    \begin{algorithm}[H]
      \KwIn{Elimination forest $\TT$, weights $\wfct \colon V \rightarrow [2n]$, target weight $\wtarget \in [2n^2]$}
      Let $r$ denote the root of $\TT$\;
      $P := \texttt{calc\_poly\_inc}(r, \emptyset)$\;
      \Return the coefficient of $Z_W^\wtarget Z_X^k$ in $P$\;
      \caption{$\countc$ for \textsc{Connected Vertex Cover}}
      \label{algo:cvc_countc}
    \end{algorithm}
    \begin{algorithm}[H]
      \KwIn{Elimination forest $\TT$, weights $\wfct \colon V \rightarrow [2n]$, vertex $v \in V$, previous choices $f \colon \tail[v] \rightarrow \{\zero, \one_L, \one_R\}$}
      \lIf{$v$ is a leaf of $\TT$}
      {
        \!\!\Return the result of equation \cref{eq:cvc_pex_leaf}\!\!
      }
      \Else
      {
        $P := 1$\;
        \For(\tcp*[f]{cf.\ equation \cref{eq:cvc_pex_branch}}){$u \in \child(v)$}
        {
          $P := P \cdot \texttt{calc\_poly\_inc}(v, f)$\;
        }
        \Return $P$\;
      }
      \caption{$\texttt{calc\_poly\_exc}(v, f)$}
      \label{algo:cvc_calc_poly_exc}
    \end{algorithm}
    \begin{algorithm}[H]
      \KwIn{Elimination forest $\TT$, weights $\wfct \colon V \rightarrow [2n]$, vertex $v \in V$, previous choices $g \colon \tail(v) \rightarrow \{\zero, \one_L, \one_R\}$}
      \For{$s \in \{\zero, \one_L, \one_R\}$}
      {
        $P_s := \texttt{calc\_poly\_exc}(v, g[v \mapsto s])$\;
      }
      \Return{$P_{\one_L} Z_W^{\wfct(v)} Z_X + P_{\one_R} Z_W^{\wfct(v)} Z_X + P_{\zero}$}\tcp*{cf.\ equation \cref{eq:cvc_pin}}
      \caption{$\texttt{calc\_poly\_inc}(v, g)$}
      \label{algo:cvc_calc_poly_inc}
    \end{algorithm}
\end{center}    
\begin{lem}
 \label{thm:cvc_count}
 Given a connected graph $G = (V, E)$, a vertex $v_1 \in V$, an integer $k$, a weight function $\wfct \colon V \rightarrow [2n]$, and an elimination forest $\TT$ of $G$ of depth $d$, we can determine $|\cutsols_\wtarget|$ modulo 2 for every $0 \leq \wtarget \leq 2n^2$ in time $\Oh^*(3^d)$ and polynomial space. In particular, \cref{algo:cvc_countc} determines $|\cutsols_\wtarget|$ modulo 2 for a specified target weight $\wtarget$ in the same time and space.
\end{lem}

\begin{proof} 
 For the discussion of the algorithm, it is convenient to drop the cardinality constraint in $\candidates$ and $\cutsols$ and to define these sets for every induced subgraph $G[V']$ of $G$. Hence, we define for every $V' \subseteq V$ the set $\widehat{\candidates}(V') = \{ X \subseteq V' \sep X \text{ is a vertex cover of } G[V'] \}$ and the set $\widehat{\cutsols}(V') = \{ (X, (X_L, X_R)) \in \cuts(V') \sep X \in \candidates(V') \text{ and } (v_1 \in V' \rightarrow v_1 \in X_L) \}$.
 
 Similar to Pilipczuk and Wrochna~\cite{PilipczukW18}, our algorithm will compute a multivariate polynomial in the formal variables $Z_W$ and $Z_X$, where the coefficient of $Z_W^\wtarget Z_X^i$ is the cardinality of $\widehat{\cutsols}_{\wtarget}^i(V) = \{(X, C) \in \widehat{\cutsols}(V) \sep \wfct(X) = \wtarget, |X| = i \}$ modulo 2, i.e., the formal variables track the weight and size of candidate solutions. In particular, we have that $\widehat{\cutsols}_{\wtarget}^k(V) = \cutsols_\wtarget$ for every $\wtarget$. Polynomials act as an appropriate data structure, because addition and multiplication of polynomials naturally updates the weight and size trackers correctly.
 
 The output polynomial is computed by a branching algorithm (see \cref{algo:cvc_countc}) that starts at the root $r$ of the elimination forest $\TT$ and proceeds downwards to the leaves. At every vertex we branch into several states, denoted $\states = \{\zero, \one_L, \one_R\}$. The interpretation of the state $\zero$ is that the vertex does not belong to the vertex cover. The states $\one_L$ and $\one_R$ indicate that the vertex is inside the vertex cover and the subscript denotes to which side of the consistent cut it belongs. 
 
 For each vertex $v$ there are multiple subproblems on $G[\broom[v]]$. When solving a subproblem, we need to take into account the choices that we have already made, i.e., the branching decisions for the ancestors of $v$. At each vertex we compute two different types of polynomials, which correspond to two different kinds of partial solutions. Those that are subsets of $\tree(v)$ and respect the choices made on $\tail[v]$ and those that are subsets of $\tree[v]$ and respect the choices made on $\tail(v)$. Distinguishing these two types of partial solutions is important when $v$ has multiple children in $\TT$. Formally, the previous branching decisions are described by assignments $f$ or $g$ from $\tail[v]$ or $\tail(v)$ to $\{\zero, \one_L, \one_R\}$ respectively.  
 
 For every vertex $v$ and assignment $f \colon \tail[v] \rightarrow \{\zero, \one_L, \one_R\}$ we define the partial solutions at $v$, but excluding $v$, that respect $f$ by
 \begin{equation*}
  \begin{aligned}
   \partialex{v}{f} = \{ & (X,(X_L, X_R)) \in \cuts(\tree(v)) \sep\, X' = X \cup f^{-1}(\{\one_L, \one_R\}), \\
                         & C' = (X_L \cup f^{-1}(\one_L), X_R \cup f^{-1}(\one_R)), 
			  (X', C') \in \widehat{\cutsols}(\broom[v]) \}.
  \end{aligned}
 \end{equation*}

 So, $\partialex{v}{f}$ consists of consistently cut subgraphs $(X, (X_L, X_R))$ of $G[\tree(v)]$ that are extended by $f$ to valid candidate-cut-pairs $(X', C')$ for $G[\broom[v]]$, meaning that $X'$ is a vertex cover of $G[\broom[v]]$ and $C'$ is a consistent cut of $G[X']$. 
 
 Very similarly, for every vertex $v$ and assignment $g \colon \tail(v) \rightarrow \{\zero, \one_L, \one_R\}$ we define the partial solutions at $v$, possibly including $v$, that respect $g$ by
 \begin{align*}
  \partialin{v}{g} = \{ & (X,(X_L, X_R)) \in \cuts(\tree[v]) \sep\, X' = X \cup g^{-1}(\{\one_L, \one_R\}), \\
                        & C' = (X_L \cup g^{-1}(\one_L), X_R \cup g^{-1}(\one_R)), 
                          (X', C') \in \widehat{\cutsols}(\broom[v]) \}.
 \end{align*}
 Thus, for the root $r$ of $\TT$ we have $\partialin{r}{\emptyset} = \widehat{\cutsols}(V)$.
 
 We keep track of the partial solutions $\partialex{v}{f}$ and $\partialin{v}{g}$ using polynomials which we define now.
 For every vertex $v$ and assignment $f \colon \tail[v] \rightarrow \{\zero, \one_L, \one_R\}$ we will compute a polynomial $\Pex{v}{f} \in \FF_2[Z_W, Z_X]$ where $\Pex{v}{f} = \sum_{\wtarget = 0}^{2n^2} \sum_{i = 0}^n c_{\wtarget, i} Z_W^\wtarget Z_X^i $ and
 \begin{align*}
  c_{\wtarget, i} = |\{(X, C) \in \partialex{v}{f} \sep \wfct(X) = \wtarget \text{ and } |X| = i\}| \mod 2.
 \end{align*}
 Similarly, for every vertex $v$ and assignment $g \colon \tail(v) \rightarrow \{\zero, \one_L, \one_R\}$ we will compute a polynomial $\Pin{v}{g} \in \FF_2[Z_W, Z_X]$ where $\Pin{v}{g} = \sum_{\wtarget = 0}^{2n^2} \sum_{i = 0}^n c'_{\wtarget, i} Z_W^\wtarget Z_X^i$ and 
 \begin{align*}
  c'_{\wtarget, i} = |\{(X, C) \in \partialin{v}{g} \sep \wfct(X) = \wtarget \text{ and } |X| = i\}| \mod 2.
 \end{align*}
 
 \cref{algo:cvc_countc} computes the polynomial $P = \Pin{r}{\emptyset}$ and extracts the appropriate coefficient of $P$. To compute $P$ we employ recurrences for $\Pex{v}{f}$ and $\Pin{v}{g}$. We proceed by describing the recurrence for $\Pex{v}{f}$.
 
 In the case that $v$ is a leaf node in $\TT$, i.e., $\tree(v) = \emptyset$, we can compute $\Pex{v}{f}$ by
 \begin{equation}
  \label{eq:cvc_pex_leaf}
  \begin{array}{lcl}
   \Pex{v}{f} & = & [f^{-1}(\{\one_L, \one_R\}) \text{ is a vertex cover of } G[\tail[v]]] \\
              & \cdot & [(f^{-1}(\one_L), f^{-1}(\one_R)) \text{ is a consistent cut of } G[f^{-1}(\{\one_L, \one_R\})]] \\
              & \cdot & [v_1 \in \tail[v] \rightarrow f(v_1) = \one_L],
  \end{array}
 \end{equation}
 which checks whether the assignment $f$ induces a valid partial solution. This is the only step in which we explicitly ensure that we are computing only vertex covers; in all other steps this will not be required. If $v$ is not a leaf, then $\Pex{v}{f}$ is computed by the recurrence
 \begin{equation}
  \label{eq:cvc_pex_branch}
  \Pex{v}{f} = \prod_{u \in \child(v)} \Pin{u}{f},
 \end{equation}
 which combines disjoint partial solutions. 
 The equations \cref{eq:cvc_pex_leaf} and \cref{eq:cvc_pex_branch} are used by \cref{algo:cvc_calc_poly_exc} to compute the polynomial $\Pex{v}{f}$. 
 
 We proceed by giving the recurrence that is used by \cref{algo:cvc_calc_poly_inc} to compute the polynomial $\Pin{v}{g}$:
 \begin{equation}
  \label{eq:cvc_pin}
  \Pin{v}{g} = \Pex{v}{g[v \mapsto \zero]} + \Pex{v}{g[v \mapsto \one_L]}\, Z_W^{\wfct(v)} Z_X + \Pex{v}{g[v \mapsto \one_R]}\, Z_W^{\wfct(v)} Z_X. 
 \end{equation}
 Equation \cref{eq:cvc_pin} tests all three possible states for $v$ in a candidate-cut-pair and multiplies by $Z_W^{\wfct(v)} Z_X$ if $v$ is in the vertex cover to update the weight and size of the partial solutions.
 
 \subparagraph*{Correctness.} We will now prove the correctness of equations \cref{eq:cvc_pex_leaf} through \cref{eq:cvc_pin}.  First of all, observe that when $v_1 \in \tail[v]$ but $f(v_1) \neq \one_L$ then we must have that $\Pex{v}{f} = 0$; similarly, we must have $\Pin{v}{g}=0$ when $g(v_1)\neq \one_L $ for $v_1\in \tail(v)$.
 This property is ensured by equation \cref{eq:cvc_pex_leaf} and preserved by the recurrences \cref{eq:cvc_pex_branch} and \cref{eq:cvc_pin}. To see that equation \cref{eq:cvc_pex_leaf} is correct, notice that when $v$ is a leaf node in $\TT$ we have that $\tree(v) = \emptyset$ and hence the only consistently cut subgraph of $\tree(v)$ is $(\emptyset, (\emptyset, \emptyset))$. Therefore, we only need to verify whether this is a valid partial solution in $\Pex{v}{f}$, which reduces to checking the predicate on the right-hand side of \cref{eq:cvc_pex_leaf}.
 
 For equations \cref{eq:cvc_pex_branch} and \cref{eq:cvc_pin}, we will establish bijections between the objects counted on either side of the respective equation and argue that size and weight are updated correctly. We proceed by proving the correctness of equation \cref{eq:cvc_pex_branch}, which is the only equation where the proof of correctness requires the special properties of elimination forests. We consider any $(X, (X_L, X_R)) \in \partialex{v}{f}$. We can uniquely partition $X$ into subsets $X^u$ of $\tree[u]$ for each $u \in \child(v)$ by setting $X^u = X \cap \tree[u]$. Furthermore, by setting $X^u_L = X_L \cap \tree[u]$ and $X^u_R = X_R \cap \tree[u]$ we obtain $(X^u, (X_L^u, X_R^u)) \in \partialin{u}{f}$, because we are only restricting the vertex cover $X' = X \cup f^{-1}(\{\one_L, \one_R\})$ and consistent cut $(X_L \cup f^{-1}(\one_L), X_R \cup f^{-1}(\one_R))$ to the induced subgraph $G[\broom[u]]$ of $G[\broom[v]]$. Vice versa, any combination of partial solutions $(X^u, (X_L^u, X_R^u)) \in \partialin{u}{f}$ for each $u \in \child(v)$ yields a partial solution $(X, (X_L, X_R)) \in \partialex{v}{f}$ as there are no edges in $G$ between $\tree[u]$ and $\tree[u']$ for $u \neq u' \in \child(v)$ by the properties of an elimination forest. Since the sets $X^u$ partition $X$, we obtain the size and weight of $X$ by summing over the sizes and weights of the sets $X^u$ respectively. Hence, these values are updated correctly by polynomial multiplication. 
 
 It remains to prove the correctness of \cref{eq:cvc_pin}. This time, consider any $(X, (X_L, X_R)) \in \partialin{v}{g}$. Now, there are three possible cases depending on the state of $v$ in this partial solution. 
 \begin{enumerate}
  \item If $v \notin X$, then we have that $(X, (X_L, X_R)) \in \partialex{v}{f}$, where $f = g[v \mapsto \zero]$. Vice versa, any $(X, (X_L, X_R)) \in \partialex{v}{f}$ must also be in $\partialin{v}{g}$. Since $X$ does not change, we do not need to update size or weight and do not multiply by further formal variables in this case. 
 
  \item If $v \in X_L \subseteq X$, then we claim that $(X \setminus \{v\}, (X_L \setminus \{v\}, X_R)) \in \partialex{v}{f}$, where $f = g[v \mapsto \one_L]$. This is true due to the identities $(X \setminus \{v\}) \cup f^{-1}(\{\one_L, \one_R\}) = X \cup g^{-1}(\{\one_L, \one_R\})$, and $(X_L \setminus \{v\}) \cup f^{-1}(\one_L) = X_L \cup g^{-1}(\one_L)$, and $X_R \cup f^{-1}(\one_R) = X_R \cup g^{-1}(\one_R)$,
  which mean that this implicitly defined mapping preserves the definition of $X'$ and $C'$ in the predicates of $\partialin{v}{g}$ and $\partialex{v}{f}$. Vice versa, any partial solution in $\partialex{v}{f}$ can be extended to such a partial solution in $\partialex{v}{g}$ by adding $v$ to $X_L$. Since $|X| - |X \setminus \{v\}| = 1$ and $\wfct(X) - \wfct(X \setminus \{v\}) = \wfct(v)$, multiplication by $Z_W^{\wfct(v)} Z_X$ updates size and weight correctly.
  
  \item If $v \in X_R \subseteq X$, the proof is analogous to case 2.
 \end{enumerate}
 If $v = v_1$, then equation \cref{eq:cvc_pin} simplifies to $\Pin{v}{g} = \Pex{v}{g[v \mapsto \one_L]} Z_W^{\wfct(v)} Z_X$, because $\partialex{v}{g[v \mapsto \zero]} = \partialex{v}{g[v \mapsto \one_R]} = \emptyset$ and hence only the second case occurs. Note that by establishing these bijections in the proofs of correctness, we have actually shown that equations \cref{eq:cvc_pex_leaf} through \cref{eq:cvc_pin} are also correct when working in $\ZZ$ instead of $\FF_2$.
 
 \subparagraph*{Time and Space Analysis.}\label{sec:cvc_time_and_space} We finish the proof by discussing the time and space requirement. Observe that the coefficients of our polynomials are in $\FF_2$ and hence can be added and multiplied in constant time. Furthermore, all considered polynomials consist of at most polynomially many monomials as the weight and size of a candidate solution are polynomial in $n$. Therefore, we can add and multiply the polynomials in polynomial time and hence compute recurrences \cref{eq:cvc_pex_leaf}, \cref{eq:cvc_pex_branch}, and \cref{eq:cvc_pin} in polynomial time. Every polynomial $\Pex{v}{f}$ and $\Pin{v}{g}$ is computed at most once, because $\Pex{v}{f}$ is only called by $\Pin{v}{g}$ where $f$ is an extension of $g$, i.e., $f = g[v \mapsto s]$ for some $s \in \states$, and $\Pin{v}{g}$ is only called by $\Pex{w}{g}$ where $w$ is the parent of $v$. Hence, the recurrences only make disjoint calls and no polynomial is computed more than once. For a fixed vertex $v$ there are at most $3^d$ choices for $f$ and $g$. Thus, \cref{algo:cvc_countc} runs in time $\Oh^*(3^d)$ for elimination forests of depth $d$. Finally, \cref{algo:cvc_countc} requires only polynomial space, because it has a recursion depth of $2d + 1$ and every recursive call needs to store at most a constant number of polynomials and each call only requires polynomial space by the previous discussion.
\end{proof}

\begin{thm}
 There is a Monte-Carlo algorithm that given an elimination forest of depth $d$ for a graph $G$ solves \textsc{Connected Vertex Cover} on $G$ in time $\Oh^*(3^d)$ and polynomial space. The algorithm cannot give false positives and may give false negatives with probability at most 1/2.
\end{thm}

\begin{proof}
 We pick an edge $\{u, v\} \in E$ and branch on $v_1 := u$ and $v_1 := v$. We run \cref{algo:cutncount} with $U = V$ and $\countc$ as given by \cref{algo:cvc_countc}. Correctness follows from \cref{thm:general_setup} and \cref{thm:cvc_correct}. The running time and space bound follow from \cref{thm:cvc_count}.
\end{proof}

We remark that calling \cref{algo:cvc_countc} for each target weight $\wtarget \in [2n^2]$ (as in \cref{algo:cutncount}) would redundantly compute the polynomial $P = \Pin{r}{\emptyset}$ several times, although it suffices to compute $P$ once and then look up the appropriate coefficient depending on $\wtarget$.

If one is interested in solving \textsc{Weighted Connected Vertex Cover}, then it is straightforward to adapt our approach to polynomially-sized weights: instead of using $Z_X$ to track the size of the vertex covers, we let it track their cost and change recurrence \cref{eq:cvc_pin} accordingly.

\subsection{Adapting the Algorithm to Other Problems}

The high-level structure of the counting procedure for the other problems is very similar to that of \cref{algo:cvc_countc} for \textsc{Connected Vertex Cover}. One possible difference is that we might have to consider the solutions over a more complicated universe $U$ than just the vertex set $V$. Also, we might want to keep track of more data of the partial solutions and hence use more than just two formal variables for the polynomials. Both of these changes occur for \textsc{Feedback Vertex Set}, which is presented in the next section. The equation for the base case (cf.\ equation \cref{eq:cvc_pex_leaf}) and the recurrence for $\Pin{v}{g}$ (cf.\ equation \cref{eq:cvc_pin}) are also problem-dependent.

\subparagraph*{Time and Space Analysis.}\label{sec:gen_time_and_space} The properties that we require of the polynomials and equations in the time and space analysis, namely that the equations can be evaluated in polynomial time and every polynomial is computed at most once, remain true by the same arguments as for \textsc{Connected Vertex Cover}. The running time essentially follows from the number of computed polynomials, which increases when we use more states for the vertices. Again denoting the set of states by $\states$, we obtain a running time of $\Oh^*(|\states|^d)$ on elimination forests of depth $d$. The space analysis also remains valid, because the recursion depth remains $2d + 1$ and for each call we need to store only a constant number of polynomials each using at most polynomial space.

\subsection{Feedback Vertex Set}
\label{sec:fvs}

\begin{problem}[doublelined]{Feedback Vertex Set}
 Input: & An undirected graph $G = (V, E)$ and an integer $k$.
 \\
 Question: & Is there a set $X \subseteq V$, $|X| = k$, such that $G - X$ is a forest?
\end{problem}

\textsc{Feedback Vertex Set} differs from the other problems in that we do not have a positive connectivity requirement, but a negative connectivity requirement, i.e., we need to ensure that the remaining graph is badly connected in the sense that it contains no cycles. Cygan et al.~\cite{CyganNPPRW11} approach this via the well-known \cref{thm:forest_lemma}.

\begin{lem}
 \label{thm:forest_lemma}
 A graph with $n$ vertices and $m$ edges is a forest if and only if it has at most $n - m$ connected components.
\end{lem}

Applying \cref{thm:forest_lemma} requires that we count how many vertices and edges remain after deleting a set $X \subseteq V$ from $G$. We do not need to count exactly how many connected components remain, it suffices to enforce that there are not too many connected components. We will achieve this, like Cygan et al.~\cite{CyganNPPRW11}, by the use of \emph{marker vertices}. In this case, our solutions are pairs $(Y, M)$ with $M \subseteq Y$, where we interpret $Y$ as the forest that remains after removing a feedback vertex set $X$ and the marked vertices are represented by the set $M$. To bound the number of connected components, we want that every connected component of $G[Y]$ contains at least one marked vertex. By forcing the marked vertices to the left side of the cut, we ensure that candidates $(Y, M)$ where $G[Y]$ has a connected component not containing a marked vertex, in particular those with more than $|M|$ connected components, cancel modulo 2. The formal definitions are $\candidates = \{ (Y, M) \sep M \subseteq Y \subseteq V \text{ and } |Y| = n - k \}$, and $\solutions = \{ (Y, M) \in \candidates \sep G[Y] \text{ is a forest, every connected component of $G[Y]$ intersects $M$} \}$, and $\cutsols = \{ ((Y, M), (Y_L, Y_R)) \sep (Y, M) \in \candidates \text{ and } (Y, (Y_L, Y_R)) \in \cuts(V) \text{ and } M \subseteq Y_L \}$.

Since our solutions $(Y, M)$ are pairs of two vertex sets, we need a larger universe to make the Isolation Lemma, \cref{thm:isolation}, work. We use $U = V \times \{\forest, \marked\}$, hence a weight function $\wfct \colon U \rightarrow [N]$ assigns two different weights $\wfct(v, \forest)$ and $\wfct(v, \marked)$ to a vertex $v$ depending on whether $v$ is marked or not. To make these definitions compatible with \cref{thm:general_setup} we associate to each pair $(Y, M)$ the set $(Y \times \{\forest\}) \cup (M \times \{\marked\}) \subseteq U$, which also allows us to extend the weight function to such pairs $(Y, M)$, i.e.\ $\wfct(Y,M) = \wfct((Y \times \{\forest\}) \cup (M \times \{\marked\}))$.

\begin{lem}[\cite{CyganNPPRW11}]
 \label{thm:fvs_cut}
 Let $(Y, M)$ be such that $M \subseteq Y \subseteq V$. The number of consistently cut subgraphs $(Y, (Y_L, Y_R))$ such that $M \subseteq Y_L$ is equal to $2^{\overline{\cc}_{M}(G[Y])}$, where $\overline{\cc}_{M}(G[Y])$ is the number of connected components of $G[Y]$ that do not contain any vertex from $M$.
\end{lem}

\begin{proof}
 For a consistently cut subgraph $(Y, (Y_L, Y_R))$ with $M \subseteq Y_L$ any connected component $C$ of $G[Y]$ that contains a vertex of $M$ must be completely contained in $Y_L$. For all other connected components $C$ of $G[Y]$, namely those counted by $\overline{\cc}_{M}(G[Y])$, we have that either $C \subseteq Y_L$ or $C \subseteq Y_R$. Thus, we arrive at the claimed number of consistently cut subgraphs $(Y, (Y_L, Y_R))$ with $M \subseteq Y_L$.
\end{proof}

To apply \cref{thm:forest_lemma}, we need to distinguish candidates by the number of edges, and markers, in addition to the weight, hence we make the following definitions for $j, \ell, \wtarget \in \NN$:
\begin{equation*}
\begin{array}{llllll}
 \candidates_\wtarget^{j, \ell} & = \{(Y, M) \in \candidates & \sep & \wfct(Y,M) = \wtarget, & |E(G[Y])| = j, & |M| = \ell\}, \\
 \solutions_\wtarget^{j, \ell} & = \{(Y, M) \in \solutions & \sep & \wfct(Y,M) = \wtarget, & |E(G[Y])| = j, & |M| = \ell\}, \\
 \cutsols_\wtarget^{j, \ell} & = \{(Y, M, (Y_L, Y_R)) \in \cutsols & \sep & \wfct(Y,M) = \wtarget, & |E(G[Y])| = j, & |M| = \ell\}.
\end{array}
\end{equation*}

\begin{lem}[\cite{CyganNPPRW11}]
 \label{thm:fvs_correct}
 Let $\wfct \colon U \rightarrow [N]$ be a weight function, and $\cutsols$ and $\solutions$ as defined above. Then we have for every $\wtarget \in \NN$ and $j \in \{0, \ldots, n - k - 1\}$ that $|\solutions_\wtarget^{j, n-k-j}| \equiv |\cutsols_\wtarget^{j, n-k-j}|$.
\end{lem}

\begin{proof}
 \cref{thm:fvs_cut} implies that $|\cutsols_\wtarget^{j, \ell}| = \sum_{(Y,M) \in \candidates_\wtarget^{j, \ell}} 2^{\overline{\cc}_{M}(G[Y])}$ for all $j, \ell, \wtarget \in \NN$. Hence, we have for all $j, \ell, \wtarget \in \NN$ that 
 \begin{equation*}
  |\cutsols_\wtarget^{j, \ell}| \equiv |\{(Y,M) \in \candidates_\wtarget^{j, \ell} \sep \overline{\cc}_{M}(G[Y]) = 0\}|.
 \end{equation*}
 We certainly have $\solutions_\wtarget^{j, \ell} \subseteq \{ (Y, M) \in \candidates_\wtarget^{j, \ell} \sep \overline{cc}_M(G[Y]) = 0\}$ by definition of $\solutions_\wtarget^{j, \ell}$. To see the other direction of inclusion for $\ell = (n - k) - j$, observe that $\cc(G[Y]) \leq |M| + \overline{\cc}_M(G[Y])$ for all $M \subseteq Y \subseteq V$, and hence a pair $(Y,M) \in \candidates_\wtarget^{j, n-k-j}$ with $\overline{\cc}_{M}(G[Y]) = 0$ must satisfy
 \begin{equation*}
  \cc(G[Y]) \leq |M| = (n-k) - j = |Y| - |E(G[Y])|.
 \end{equation*}
 Finally, \cref{thm:forest_lemma} implies that $G[Y]$ is a forest and this finishes the other direction of inclusion. Thus, we have that $|\cutsols_\wtarget^{j, n-k-j}| \equiv |\solutions_\wtarget^{j, n-k-j}|$.
\end{proof}

Note that by \cref{thm:forest_lemma} a \textsc{Feedback Vertex Set} instance has a solution $X$ if and only if there is a choice of $\wtarget \in \NN$, $j \in \{0\} \cup [n - k - 1]$ such that $\solutions_\wtarget^{j, n-k-j} \neq \emptyset$.

\begin{lem}
 \label{thm:fvs_count}
 Given a connected graph $G = (V, E)$, an integer $k$, a weight function $\wfct \colon U \rightarrow [4n]$ and an elimination forest $\TT$ of $G$ of depth $d$, we can determine $|\cutsols_\wtarget^{j, n - k - j}|$ modulo 2 for every $0 \leq \wtarget \leq 4n^2$, $0 \leq j \leq \min(m, n-k-1)$, in time $\Oh^*(3^d)$ and polynomial space.
\end{lem}

\begin{proof}
 Again, we drop the cardinality constraints from $\candidates$ and $\cutsols$ and define for induced subgraphs $G[V']$ the variants $\widehat{\candidates}(V') = \{ (Y, M) \sep M \subseteq Y \subseteq V' \}$ and $\widehat{\cutsols}(V') = \{ ((Y, M), (Y_L, Y_R)) \sep (Y, M) \in \widehat{\candidates}(V') \text{ and } (Y, (Y_L, Y_R)) \in \cuts(V') \text{ and } M \subseteq Y_L\}$.
 
 We will compute a multivariate polynomial in the formal variables $Z_W, Z_Y, Z_E, Z_M$, where the coefficient of $Z_W^\wtarget Z_Y^i Z_E^j Z_M^\ell$ is the cardinality modulo 2 of
 \begin{equation*}
  \widehat{\cutsols}_{\wtarget}^{i, j, \ell} = \{((Y,M),C) \in \widehat{\cutsols}(V) \sep \wfct(Y,M) = \wtarget, |Y| = i, |E(G[Y])| = j, |M| = \ell \}.
 \end{equation*}
 The coefficients of $Z_W^\wtarget Z_Y^{n - k} Z_E^j Z_M^{n-k-j}$ for every $\wtarget$ and $j$ then yield the desired numbers.
 
 For \textsc{Feedback Vertex Set} we require three states which are given by $\states = \{\zero_L, \zero_R, \one\}$. The states $\zero_L$ and $\zero_R$ represent vertices inside the remaining forest and the subscript denotes to which side of the consistent cut a vertex belongs; the state $\one$ represents vertices inside the feedback vertex set. Perhaps surprisingly, there is no state to represent marked vertices. It turns out that it is not important which vertices are marked; it is sufficient to know the number of marked vertices.
 
 For every vertex $v$ and assignment $f \colon \tail[v] \rightarrow \{\zero_L, \zero_R, \one\}$, we define the partial solutions at $v$, but excluding $v$, that respect $f$ by
 \begin{align*}
  \partialex{v}{f} = \{ & ((Y,M), (Y_L, Y_R)) \in \widehat{\cutsols}(\tree(v)) \sep\, Y' = Y \cup f^{-1}(\{\zero_L, \zero_R\}), \\
			& C' = (Y_L \cup f^{-1}(\zero_L), Y_R \cup f^{-1}(\zero_R)), 
			  ((Y', M), C') \in \widehat{\cutsols}(\broom[v]) \}.
 \end{align*}
 The partial solutions in $\partialex{v}{f}$ are consistently cut subgraphs $(Y, (Y_L, Y_R))$ of $G[\tree(v)]$ where a subset $M$ of the left side is marked and the extension to $(Y', C')$ by $f$ is a consistently cut subgraph of $G[\broom[v]]$.
 
 Similarly, for every vertex $v$ and assignment $g \colon \tail(v) \rightarrow \{\zero_L, \zero_R, \one\}$, we define the partial solutions at $v$, possibly including $v$, that respect $g$ by
 \begin{align*}
  \partialin{v}{g} = \{ & ((Y,M), (Y_L, Y_R)) \in \widehat{\cutsols}(\tree[v]) \sep\, Y' = Y \cup g^{-1}(\{\zero_L, \zero_R\}), \\
			& C' = (Y_L \cup g^{-1}(\zero_L), Y_R \cup g^{-1}(\zero_R)), 
			  ((Y', M), C') \in \widehat{\cutsols}(\broom[v]) \}.
 \end{align*}
 
 For every vertex $v$ and assignment $f \colon \tail[v] \rightarrow \{\zero_L, \zero_R, \one\}$, we will compute a polynomial $\Pex{v}{f} \in \FF_2[Z_W, Z_Y, Z_E, Z_M]$ where the coefficient of $Z_W^\wtarget Z_Y^i Z_E^j Z_M^\ell$ in $\Pex{v}{f}$ is given by
 \begin{align*}
  |\{ ((Y,M), C) \in \partialex{v}{f} \sep \quad & \wfct(Y,M) = \wtarget, |Y| = i, \\
   & |E(G[Y])| + |E(Y,f^{-1}(\{\zero_L, \zero_R\}))| = j, |M| = \ell\}| \mod 2.
 \end{align*}
 For every vertex $v$ and assignment $g \colon \tail(v) \rightarrow \{\zero_L, \zero_R, \one\}$, we will compute a polynomial $\Pin{v}{g} \in \FF_2[Z_W, Z_Y, Z_E, Z_M]$ where the coefficient of $Z_W^\wtarget Z_Y^i Z_E^j Z_M^\ell$ in $\Pin{v}{g}$ is given by
 \begin{align*}
  |\{ ((Y,M), C) \in \partialin{v}{g} \sep \quad & \wfct(Y,M) = \wtarget, |Y| = i, \\
   & |E(G[Y])| + |E(Y,g^{-1}(\{\zero_L, \zero_R\}))| = j, |M| = \ell\}| \mod 2.
 \end{align*}
 The exponents of the monomials $Z_W^\wtarget Z_Y^i Z_E^j Z_M^\ell$ in $\Pex{v}{f}$ and $\Pin{v}{g}$ range between $0 \leq \wtarget \leq 4n^2$, $0 \leq i \leq n$, $0 \leq j \leq m$, and $0 \leq \ell \leq n$. Observe that the term $|E(G[Y])| + |E(Y,g^{-1}(\{\zero_L, \zero_R\}))|$ degenerates to $|E(G[Y])|$ when $v = r$ is the root node. Hence, $\Pin{r}{\emptyset}$ yields the desired polynomial.
 
 We now present the recurrences used to compute the polynomials $\Pex{v}{f}$ and $\Pin{v}{g}$. If $v$ is a leaf node in $\TT$, then we can compute $\Pex{v}{f}$ by
 \begin{equation}
  \label{eq:fvs_pex_leaf}
  \Pex{v}{f} = [(f^{-1}(\zero_L), f^{-1}(\zero_R)) \text{ is a consistent cut of } G[f^{-1}(\{\zero_L, \zero_R\})]].
 \end{equation}
 If $v$ is not a leaf node, then we compute $\Pex{v}{f}$ by
 \begin{equation}
  \label{eq:fvs_pex_branch}
  \Pex{v}{f} = \prod_{u \in \child(v)} \Pin{u}{f}.
 \end{equation}
 To compute $\Pin{v}{g}$ we use the recurrence
 \begin{equation}
  \label{eq:fvs_pin}
 \begin{array}{llllll}
  \Pin{v}{g} 	& = \Pex{v}{g[v \rightarrow \zero_L]} & Z_W^{\wfct(v,\forest)} & Z_Y & Z_E^{|N(v) \cap g^{-1}(\{\zero_L, \zero_R\})|} & \\
		& + \,\Pex{v}{g[v \rightarrow \zero_L]} & Z_W^{\wfct(v,\forest) + \wfct(v,\marked)} & Z_Y & Z_E^{|N(v) \cap g^{-1}(\{\zero_L, \zero_R\})|} & Z_M \\
		& + \,\Pex{v}{g[v \rightarrow \zero_R]} & Z_W^{\wfct(v,\forest)} & Z_Y & Z_E^{|N(v) \cap g^{-1}(\{\zero_L, \zero_R\})|}. & \\
		& + \,\Pex{v}{g[v \rightarrow \one]} & & & & \\
 \end{array}  
 \end{equation}
 This recurrence tests all three possible states for the vertex $v$ and whether it is marked. In the second case $v$ has state $\zero_L$ and is marked, so the formal variables $Z_W$ and $Z_M$ must be updated differently from the case where $v$ is not marked.
 
 We will now prove the correctness of the equations \cref{eq:fvs_pex_leaf} to \cref{eq:fvs_pin}. For the correctness of equation \cref{eq:fvs_pex_leaf}, notice that $\tree(v) = \emptyset$ when $v$ is a leaf. Hence, $\widehat{\cutsols}(\tree(v))$ degenerates to $\{ ((\emptyset, \emptyset), (\emptyset, \emptyset)) \}$ and we must check whether $((\emptyset, \emptyset), (\emptyset, \emptyset)) \in \partialex{v}{f}$ which means that $((f^{-1}(\{\zero_L, \zero_R\}), \emptyset), (f^{-1}(\zero_L), f^{-1}(\zero_R))) \in \widehat{\cutsols}(\broom[v])$ and checking the consistency of the cut is the only nontrivial requirement in this case and all trackers are zero. 
 
 The proof of correctness for equation \cref{eq:fvs_pex_branch} is similar to the proof for equation \cref{eq:cvc_pex_branch} of \textsc{Connected Vertex Cover}. Any solution $((Y, M), C) \in \partialex{v}{f}$ uniquely partitions into solutions $((Y^u, M^u), C^u) \in \partialin{u}{f}$ for each $u \in \child(v)$. Vice versa, any combination of solutions for the children $u$ of $v$ yields a unique solution in $\partialex{v}{f}$. The properties of an elimination forest are needed to show that the union of consistent cuts remains a consistent cut. Observe that the sets $E(Y^u, f^{-1}(\{\zero_L, \zero_R\}))$, $u \in \child(v)$, must be pairwise disjoint as well, hence polynomial multiplication also updates the edge tracker correctly. We omit further details.
 
 To prove the correctness of equation \cref{eq:fvs_pin} we consider a partial solution $((Y, M), (Y_L, Y_R)) \in \partialin{v}{g}$ and distinguish between four cases depending on the state of $v$.
 \begin{enumerate}
  \item If $v \in Y_L \subseteq Y$ and $v \notin M$, then $((Y \setminus \{v\}, M),(Y_L \setminus \{v\}, Y_R)) \in \partialex{v}{f}$, where $f = g[v \mapsto \zero_L]$, because the definition of $Y'$ and $C'$ in the predicate in the definition of $\partialin{v}{g}$ and $\partialex{v}{f}$ do not change. Hence, we can also extend any partial solution of $\partialex{v}{f}$ to such a partial solution of $\partialin{v}{g}$ by adding $v$ to $Y_L$. The number of vertices in $Y$ increase by $1$ and the weight increases by $\wfct(v, \forest)$. Let $Y' = Y \setminus \{v\}$ and $A = g^{-1}(\{\zero_L, \zero_R\})$, then $A \cup \{v\} = f^{-1}(\{\zero_L, \zero_R\})$ and consider the following computation to see that the edge count increases by $|N(v) \cap g^{-1}(\{\zero_L, \zero_R\})|$:
  \begin{align*}
		  & (|E(G[Y])| - |E(G[Y'])| ) + |E(Y,A)| - |E(Y',A \cup \{v\})| \\
		  = \quad & |E(Y', \{v\})| + (|E(Y', A)| + |E(\{v\}, A)|) - (|E(Y', A)| + |E(Y', \{v\})|) \\
		  = \quad & |E(\{v\}, A)| = |N(v) \cap g^{-1}(\{\zero_L, \zero_R\})|.
  \end{align*}  
   Therefore, multiplication with $ Z_W^{\wfct(v,\forest)} Z_Y Z_E^{|N(v) \cap g^{-1}(\{\zero_L, \zero_R\})|}$ is the correct update.
  \item If $v \in M \subseteq Y_L \subseteq Y$, then $((Y \setminus \{v\}, M \setminus \{v\}),(Y_L \setminus \{v\}, Y_R)) \in \partialex{v}{f}$, where $f = g[v \mapsto \zero_L]$. The argument is similar to case 1. Note that, again, the definitions of $Y'$ and $C'$ do not change in the predicates. The set of marked vertices $M$ does change, but we only need to ensure that $M$ remains a subset of the left side of the cut, which we do by removing $v$ from $M$. In addition to the changes in the number of vertices and edges from case 1, the number of marked vertices has increased by $1$ and the weight increases by an additional $\wfct(v, \marked)$, to keep track of these changes we further multiply by $Z_W^{\wfct(v, \marked)} Z_M$.
  \item If $v \in Y_R \subseteq Y$, then the proof is analogous to case 1.
  \item If $v \notin Y$, then $((Y,M),(Y_L,Y_R)) \in \partialex{v}{f}$, where $f = g[v \mapsto \one]$, because there is no constraint involving vertices with state $\one$. Vice versa, we have that any partial solution $((Y,M),(Y_L,Y_R)) \in \partialex{v}{f}$ must also be in $\partialin{v}{g}$. Since $Y$ and $M$ do not change, we do not need to multiply by further formal variables.
 \end{enumerate}
 The running time and space bound follows from the general discussion in \cref{sec:gen_time_and_space}.
\end{proof}

\begin{thm}
 \label{thm:fvs_algo}
 There exists a Monte-Carlo algorithm that given an elimination forest of depth $d$ solves \textsc{Feedback Vertex Set} in time $\Oh^*(3^d)$ and polynomial space. The algorithm cannot give false positives and may give false negatives with probability at most 1/2.
\end{thm}

\begin{proof} 
 We set $U = V \times \{\forest, \marked\}$, but we need to slightly adapt the definition of $\solutions$ and $\cutsols$ to be able to apply \cref{thm:general_setup}. We define $\widetilde{\solutions} = \cup_{j = 0}^{n-k-1} \cup_{\wtarget = 0}^{4n^2} \solutions_w^{j,n-k-j}$ and $\widetilde{\cutsols} = \cup_{j = 0}^{n-k-1} \cup_{\wtarget = 0}^{4n^2} \cutsols_w^{j, n-k-j}$. Note that $\solutions$ is nonempty if and only if $\widetilde{\solutions}$ is nonempty by \cref{thm:forest_lemma}. The procedure $\countc$ is given by running the algorithm from \cref{thm:fvs_count} and for a given target weight $\wtarget$ adding up (modulo 2) the values of $|\cutsols_\wtarget^{j, n-k-j}|$ for $j = 0, \ldots, n-k-1$, thereby obtaining the cardinality of $\widetilde{\cutsols}_\wtarget$ modulo 2. The desired algorithm is then given by running \cref{algo:cutncount}. The correctness follows from \cref{thm:fvs_correct} and \cref{thm:general_setup} with $\widetilde{\solutions}$ and $\widetilde{\cutsols}$ instead of $\solutions$ and $\cutsols$. The running time and space bound follows from \cref{thm:fvs_count}.
\end{proof}

\cref{thm:fvs_algo} allows us to easily reobtain a result by Cygan et al.~\cite{CyganNPPRW11} on \textsc{Feedback Vertex Set} parameterized by \textsc{Feedback Vertex Set}. Recently, this result has been superseded by results of Li and Nederlof~\cite{LiN20}; they present an $\Oh^*(2.7^k)$-time and exponential-space algorithm and an $\Oh^*(2.8446^k)$-time and polynomial-space algorithm for this problem.

\begin{cor}[\cite{CyganNPPRW11}]
 \label{thm:fvs_fvs}
 There is a Monte-Carlo algorithm that given a feedback vertex set of size $s$ solves \textsc{Feedback Vertex Set} in time $\Oh^*(3^s)$ and polynomial space. The algorithm cannot give false positives and may give false negatives with probability at most 1/2.
\end{cor}

\begin{proof}
 Given a feedback vertex set $X'$ of size $s$, we can construct an elimination forest $\TT$ of depth $s + \Oh(\log n)$ by arranging $X'$ in a path and attaching an optimal elimination forest for $G - X'$ below this path. Since $G - X'$ is a forest, we know that $\td(G - X') \in \Oh(\log n)$ and we can compute an optimal elimination forest for $G - X'$ in polynomial time (see, e.g., \cite{NesetrilO12}), hence we obtain an elimination forest $\TT$ of the desired depth. Now, we apply \cref{thm:fvs_algo} with the elimination forest $\TT$ and obtain the desired running time as $\Oh^*(3^{s + \Oh(\log n)}) = \Oh^*(3^s \mathrm{poly}(n)) = \Oh^*(3^s)$. 
\end{proof}

\subsection{Connected Dominating Set}

\begin{problem}[doublelined]{Connected Dominating Set}
 Input: & An undirected graph $G = (V, E)$ and an integer $k$.
 \\
 Question: & Is there a set $X \subseteq V$, $|X| = k$, such that $G[X]$ is connected and $X$ is a dominating set of $G$, i.e., $N[X] = V$?
\end{problem}

Obtaining a fast polynomial-space algorithm for \textsc{Connected Dominating Set} is interesting, because already the algorithm running in time $\Oh^*(3^{\td(G)})$ and polynomial space for \textsc{Dominating Set} by Pilipczuk and Wrochna~\cite{PilipczukW18} is nontrivial. Their algorithm actually counts dominating sets  by recursively applying a small inclusion-exclusion formula. Combining inclusion-exclusion with branching has appeared in the literature before and is also called inclusion-exclusion-branching \cite{Bax93, NederlofR10, NederlofRD14}. Unlike for \textsc{Feedback Vertex Set} we do not need any further ideas to apply the Cut\&Count technique, and it turns out that applying the inclusion-exclusion approach simultaneously is not an issue.

As for \textsc{Connected Vertex Cover}, we have that $U = V$. Our solution candidates are given by $\candidates = \{ X \subseteq V \sep X \text{ is a dominating set of } G \text{ and } |X| = k \}$, and our solutions are given by $\solutions = \{ X \in \candidates \sep G[X] \text{ is connected} \}$. Again, we need to fix a vertex $v_1$ to the left side $X_L$ of the cut, so that we do not have an even number of consistent cuts for connected candidates. We will iterate over all possible choices of $v_1$. Hence, the formal definition of the candidate-cut-pairs for \textsc{Connected Dominating Set} is
\begin{equation*}
 \cutsols = \{ (X, (X_L, X_R)) \in \cuts(V) \sep X \in \candidates \text{ and } v_1 \in X_L\}.
\end{equation*}

The proof of correctness for these definitions is completely analogous to the proof of \cref{thm:cvc_correct}. 

\begin{lem}[\cite{CyganNPPRW11}]
 \label{thm:cds_correct}
 Let $\wfct \colon V \rightarrow [N]$ be a weight function, and $\cutsols$ and $\solutions$ as defined above.
 Then we have for every $\wtarget \in \NN$ that $|\solutions_\wtarget| \equiv |\cutsols_\wtarget|$.
\end{lem}

\begin{proof}
 Similar to the proof of \cref{thm:cvc_correct}.
\end{proof}

We proceed with describing the procedure $\countc$ for the case of \textsc{Connected Dominating Set}.

\begin{lem}
 \label{thm:cds_count}
 Given a connected graph $G = (V,E)$, a vertex $v_1 \in V$, an integer $k$, a weight function $\wfct \colon V \rightarrow [2n]$ and an elimination forest $\TT$ of $G$ of depth $d$, we can determine $|\cutsols_\wtarget|$ modulo 2 for every $0 \leq \wtarget \leq 2n^2$ in time $\Oh^*(4^d)$ and polynomial space.
\end{lem}

\begin{proof}
 Similar to \textsc{Connected Vertex Cover} we compute a multivariate polynomial in the formal variables $Z_W$ and $Z_X$, where the coefficient of $Z_W^\wtarget Z_X^i$ is the cardinality of $\cutsols_{\wtarget}^i = \{(X,C) \in \cutsols_\wtarget \sep |X| = i\}$ modulo 2. This time, we require four different states and they are denoted $\states = \{\zero_A, \zero_F, \one_L, \one_R\}$. The states $\zero_A$ and $\zero_F$ are used for vertices not in the dominating set and related to the inclusion-exclusion approach; vertices with the state $\zero_A$ are allowed to be dominated by the considered partial solutions, whereas vertices with the state $\zero_F$ are forbidden, i.e.\ not allowed, to be dominated by the considered partial solution. One can then obtain the number of solutions that dominate a given vertex $v$ by computing the number of partial solutions where $v$ is allowed and subtracting the number of partial solutions where $v$ is forbidden. The interpretation of the states $\one_L$ and $\one_R$ is to include the vertex into the dominating set and the subscript denotes the side of the consistent cut.  
 
 Although we want to compute polynomials in $\FF_2[Z_W, Z_X]$, working in $\FF_2$ would obfuscate some details of the inclusion-exclusion approach. Hence, we choose to present the algorithm so that it computes polynomials in $\ZZ[Z_W, Z_X]$. We can apply the ring homomorphism $\ZZ \rightarrow \FF_2$, $x \mapsto x \mod 2$, to also obtain correctness for the case of $\FF_2$.
 
 For every vertex $v$ and assignment $f \colon \tail[v] \rightarrow \{\zero_A, \zero_F, \one_L, \one_R\}$ we define the partial solutions at $v$, but excluding $v$, that respect $f$ by
 \begin{align*}
  \partialex{v}{f} = \{ (X,(X_L, X_R)) \in \cuts(\tree(v)) \sep\, & X' = X \cup f^{-1}(\{\one_L, \one_R\}), \\
								  & C' = (X_L \cup f^{-1}(\one_L), X_R \cup f^{-1}(\one_R)), \\
								  & (X', C') \in \cuts(\broom[v]), \\
								  & \tree(v) \subseteq N[X'], \,\, N[X'] \cap f^{-1}(\zero_F) = \emptyset, \\
								  & v_1 \in \broom[v] \rightarrow v_1 \in X_L \cup f^{-1}(\one_L) \}.
 \end{align*}
 As usual, our partial solutions are consistently cut subgraphs of $\tree(v)$ that are extended to consistently cut subgraphs $(X', C')$ of $\broom[v]$ by the assignment $f$. The penultimate line states that $X'$ should dominate $\tree(v)$ completely and not dominate any forbidden vertex, the latter part is crucial for the inclusion-exclusion-branching. The last line ensures that $v_1$ is on the left side of the consistent cut. Note that unlike for \textsc{Connected Vertex Cover}, the extensions of the partial solutions do not need to be dominating sets of $G[\broom[v]]$ due to the inclusion-exclusion approach.
 
 Similarly, for every vertex $v$ and assignment $g \colon \tail(v) \rightarrow \{\zero_A, \zero_F, \one_L, \one_R\}$ we define the partial solutions at $v$, possibly including $v$, that respect $g$ by
 \begin{align*}
  \partialin{v}{g} = \{ (X,(X_L, X_R)) \in \cuts(\tree[v]) \sep\, & X' = X \cup g^{-1}(\{\one_L, \one_R\}), \\
								  & C' = (X_L \cup g^{-1}(\one_L), X_R \cup g^{-1}(\one_R)), \\
								  & (X', C') \in \cuts(\broom[v]), \\
								  & \tree[v] \subseteq N[X'],\,\, N[X'] \cap g^{-1}(\zero_F) = \emptyset, \\
								  & v_1 \in \broom[v] \rightarrow v_1 \in X_L \cup g^{-1}(\one_L) \}.
 \end{align*}
 
 For every vertex $v$ and assignment $f \colon \tail[v] \rightarrow \{\zero_A, \zero_F, \one_L, \one_R\}$ we will compute a polynomial $\Pex{v}{f} \in \ZZ[Z_W, Z_X]$ where the coefficient of $Z_W^\wtarget Z_X^i$ in $\Pex{v}{f}$ is given by
 \begin{equation*}
  |\{(X, C) \in \partialex{v}{f} \sep \wfct(X) = \wtarget \text{ and } |X| = i\}|.
 \end{equation*}
 Similarly, for every vertex $v$ and assignment $g \colon \tail(v) \rightarrow \{\zero_A, \zero_F, \one_L, \one_R\}$ we will compute a polynomial $\Pin{v}{g} \in \ZZ[Z_W, Z_X]$ where the coefficient of $Z_W^\wtarget Z_X^i$ in $\Pin{v}{g}$ is given by
 \begin{equation*}
  |\{(X, C) \in \partialin{v}{g} \sep \wfct(X) = \wtarget \text{ and } |X| = i\}|.
 \end{equation*}
 The exponents of the monomials $Z_W^\wtarget Z_X^i$ range between $0 \leq \wtarget \leq 2n^2$ and $0 \leq i \leq n$. By extracting the coefficients of the monomials $Z_W^\wtarget Z_X^k$ in $\Pin{r}{\emptyset}$ for $0 \leq \wtarget \leq 2n^2$, where $r$ is the root of $\TT$, we obtain the desired numbers.
 
 We will now state the recurrences used to compute these polynomials. If $v$ is a leaf node, then we compute $\Pex{v}{f}$ by
 \begin{equation}
  \label{eq:cds_pex_leaf}
 \begin{array}{lcl}
  \Pex{v}{f} & =     & [(f^{-1}(\{\one_L, \one_R\}), (f^{-1}(\one_L), f^{-1}(\one_R)) \in \cuts(\tail[v])] \\
	     & \cdot & [N[f^{-1}(\{\one_L, \one_R\})] \cap f^{-1}(\zero_F) = \emptyset] \\
	     & \cdot & [v_1 \in \tail[v] \rightarrow f(v_1) = \one_L],
 \end{array} 
 \end{equation}
 where the first line checks that $f$ defines a consistent cut, the second line checks that no forbidden vertex is dominated, and the last line ensures that $v_1$ is on the left side of the consistent cut.
 
 Otherwise, we compute $\Pex{v}{f}$ by
 \begin{equation}
  \label{eq:cds_pex_branch}
  \Pex{v}{f} = \prod_{u \in \child(v)} \Pin{u}{f}.
 \end{equation}
 To compute $\Pin{v}{g}$ we use the recurrence
 \begin{equation}
 \begin{aligned}
  \label{eq:cds_pin}
  \Pin{v}{g}  & = (\Pex{v}{g[v \mapsto \zero_A]} - \Pex{v}{g[v \mapsto \zero_F]}) \\
	         & \, + (\Pex{v}{g[v \mapsto \one_L]}  + \Pex{v}{g[v \mapsto \one_R]}) Z_W^{\wfct(v)} Z_X.
 \end{aligned} 
 \end{equation}
 In a consistently cut dominating set the vertex $v$ has three possible states: $v$ can be on the left or right side of the cut of the dominating set, or $v$ is not in the dominating set and hence dominated by some other vertex. The first two summands in equation \cref{eq:cds_pin} correspond to the first two states and to count the solutions that dominate $v$ we use an inclusion-exclusion formula, i.e., we count the partial solutions that are allowed to dominate $v$ and subtract the partial solutions that do not dominate $v$.
 
 We proceed by proving the correctness of equations \cref{eq:cds_pex_leaf} to \cref{eq:cds_pin}. As for \textsc{Connected Vertex Cover}, the polynomials $\Pex{v}{f}$ and $\Pin{v}{g}$ vanish when $v_1$ is assigned a state different from $\one_L$. Equation \cref{eq:cds_pex_leaf} is correct, because we have $\tree(v) = \emptyset$ when $v$ is a leaf and hence the predicate in the definition of $\partialex{v}{f}$ degenerates to the predicate that is checked in equation \cref{eq:cds_pex_leaf}. We omit the proof of correctness for equation \cref{eq:cds_pex_branch} as it is very similar to the proof of correctness for equation \cref{eq:cvc_pex_branch} in the case of \textsc{Connected Vertex Cover}, where the properties of an elimination forest are required to show that the union of consistent cuts remains a consistent cut and that vertices from different subtrees cannot dominate each other.
 
 To prove the correctness of equation \cref{eq:cds_pin}, consider any $(X, (X_L, X_R)) \in \partialin{v}{g}$. We distinguish three cases depending on the state of $v$.
 \begin{enumerate}
 \item If $v \notin X$, then $v$ must be dominated by a vertex in $X' = X \cup g^{-1}(\{\one_L, \one_R\})$. First of all, notice that $\partialex{v}{g[v \mapsto \zero_F]} \subseteq \partialex{v}{g[v \mapsto \zero_A]}$. Hence, it suffices to argue that $(X, (X_L, X_R)) \in \partialex{v}{g[v \mapsto \zero_A]} \setminus \partialex{v}{g[v \mapsto \zero_F]}$. Certainly, $(X, (X_L, X_R)) \in \partialex{v}{g[v \mapsto \zero_A]}$ as the constraints of $\partialin{v}{g}$ are more strict than the constraints of $\partialex{v}{g[v \mapsto \zero_A]}$. Due to $v \in \tree[v] \subseteq N[X']$ we have that $(X, (X_L, X_R)) \notin \partialex{v}{f}$ where $f = g[v \mapsto \zero_F]$, because the constraint $N[X'] \cap f^{-1}(\zero_F) = \emptyset$ is violated by $v \in N[X'] \cap f^{-1}(\zero_F)$. Vice versa, consider any $(X, (X_L, X_R)) \in \partialex{v}{g[v \mapsto \zero_A]} \setminus \partialex{v}{g[v \mapsto \zero_F]}$, this also implies that $v$ must be dominated by a vertex in $X'$ and hence we have $(X, (X_L, X_R)) \in \partialin{v}{g}$. Since $X$ remains unchanged, there is no need to update the formal variables $Z_W$ and $Z_X$.
  \item If $v \in X_L \subseteq X$, then $(X \setminus \{v\}, (X_L \setminus \{v\}, X_R)) \in \partialex{v}{f}$ where $f = g[v \mapsto \one_L]$, because the definition of $X'$ and $C'$ in $\partialin{v}{g}$ and $\partialex{v}{f}$ remain unchanged. Vice versa, any partial solution in $\partialex{v}{f}$ can be extended to such a solution of $\partialin{v}{g}$. The only constraint where a set changes is the change from $\tree(v) \subseteq N[X']$ to $\tree[v] \subseteq N[X']$, but since $v \in X'$ the former implies the latter. Since $v$ was added to the partial solution, the weight is increased by $\wfct(v)$ and the size is increased by $1$; we keep track of this change by multiplying with $Z_W^{\wfct(v)} Z_X$.  
  \item If $v \in X_R \subseteq X$, then the proof is analogous to case 2.
 \end{enumerate}

 The running time and space bound follows from the general discussion in \cref{sec:gen_time_and_space}.
\end{proof}

\begin{thm}
 There is a Monte-Carlo algorithm that given an elimination forest of depth $d$ solves \textsc{Connected Dominating Set} in time $\Oh^*(4^d)$ and polynomial space. The algorithm cannot give false positives and may give false negatives with probability at most 1/2.
\end{thm}

\begin{proof}  
 We iterate over all vertices $v \in V$ and set $v_1 := v$. We then run \cref{algo:cutncount} with $U = V$ and the procedure $\countc$ as given by \cref{thm:cds_count}. The correctness follows from \cref{thm:general_setup} and \cref{thm:cds_correct}. The running time and space bound follow from \cref{thm:cds_count}. 
\end{proof}

\subsection{Steiner Tree}

\begin{problem}[doublelined]{Steiner Tree}
 Input: & An undirected graph $G = (V, E)$, a set of terminals $T \subset V$, and an integer $k$. \\
 Question: & Is there a set $X \subseteq V$, $|X| = k$, such that $G[X]$ is connected and $T \subseteq V$?
\end{problem}

It is convenient for us to choose vertices instead of edges in \textsc{Steiner Tree}, to obtain a Steiner tree with $k$ edges one can take a solution $X$ to \textsc{Steiner Tree} of size $k + 1$ and then compute a spanning tree of $G[X]$ in polynomial time. \textsc{Steiner Tree} is very similar to \textsc{Connected Vertex Cover}, indeed only the equation for the base case is different. We choose an arbitrary terminal $t_1 \in T$ and fix it to the left side of the consistent cuts, hence we obtain the following definitions:
\begin{align*}
 \candidates & = \{ X \subseteq V \sep T \subseteq V, |T| = k \}, \\
 \solutions & = \{X \in \candidates \sep G[X] \text{ is connected}\}, \\
 \cutsols & = \{(X, (X_L, X_R)) \in \cuts(V) \sep X \in \candidates \text{ and } t_1 \in X_L \}.
\end{align*}

With $v_1 = t_1$, the correctness of these definitions, in regards to the Cut\&Count technique, is proved exactly as for \textsc{Connected Vertex Cover}.  

\begin{lem}[\cite{CyganNPPRW11}]
 \label{thm:st_correct}
 Let $\wfct \colon V \rightarrow [N]$ be a weight function, and $\cutsols$ and $\solutions$ as defined above. Then we have for every $\wtarget \in \NN$ that $|\solutions_\wtarget| \equiv |\cutsols_\wtarget|$.
\end{lem}

\begin{proof}
 Similar to the proof of \cref{thm:cvc_correct}.
\end{proof}

We present the counting procedure for \textsc{Steiner Tree} now.

\begin{lem}
 \label{thm:st_count}
 Given a connected graph $G = (V, E)$, terminals $T \subseteq V$, an integer $k$, a weight function $\wfct \colon V \rightarrow [2n]$, and an elimination forest $\TT$ of depth $d$ for $G$, we can determine $|\cutsols_\wtarget|$ modulo 2 for all $0 \leq \wtarget \leq 2n^2$ in time $\Oh^*(3^d)$ and polynomial space.
\end{lem}

\begin{proof}
 As usual, we drop the cardinality constraint and define the candidates and candidate-cut-pairs for induced subgraphs $G[V']$, where $V' \subseteq V$, by 
 \begin{align*}
  \widehat{\candidates}(V') & = \{ X \subseteq V' \sep T \cap V' \subseteq X\} \\
  \widehat{\cutsols}(V') & = \{ (X, (X_L, X_R)) \in \cuts(V') \sep X \in \candidates(V') \text{ and } t_1 \in V' \Rightarrow t_1 \in X_L \}.
 \end{align*}
 The states are given by $\states = \{\zero, \one_L, \one_R\}$ with the usual interpretation. For every vertex $v$ and assignment $f \colon \tail[v] \rightarrow \{\zero,\one_L, \one_R\}$ we define the partial solutions at $v$, but excluding $v$, that respect $f$ by
 \begin{align*}
  \partialex{v}{f} = \{ (X, (X_L, X_R)) \in \widehat{\cutsols}(\tree(v)) \sep\, & X' = X \cup f^{-1}(\{\one_L, \one_R\}), \\
                                                                & C' = (X_L \cup f^{-1}(\one_L), X_R \cup f^{-1}(\one_R)), \\
                                                                & (X', C') \in \widehat{\cutsols}(\broom[v]) \}.
 \end{align*}
 So, the partial solutions in $\partialex{v}{f}$ are given by consistently cut subgraphs of $G[\tree(v)]$ that are extended to candidate-cut-pairs for $G[\broom[v]]$ by $f$, i.e.\ consistently cut subgraphs of $G[\broom[v]]$ that contain all terminals in $\broom[v]$. The coefficient of $Z_W^\wtarget Z_X^i$ in the corresponding polynomial $\Pex{v}{f} \in \FF_2[Z_W, Z_X]$ is given by
 \begin{align*}
  |\{(X, C) \in \partialex{v}{f} \sep \wfct(X) = \wtarget \text{ and } |X| = i\}| \mod 2.
 \end{align*}

 Similarly, for every vertex $v$ and assignment $g \colon \tail(v) \rightarrow \{\zero, \one_L, \one_R\}$ we define the partial solutions at $v$, possibly including $v$, that respect $g$ by
 \begin{align*}
  \partialin{v}{g} = \{ (X, (X_L, X_R)) \in \widehat{\cutsols}(\tree[v]) \sep\, & X' = X \cup g^{-1}(\{\one_L, \one_R\}), \\
                                                                & C' = (X_L \cup g^{-1}(\one_L), X_R \cup g^{-1}(\one_R)), \\
                                                                & (X', C') \in \widehat{\cutsols}(\broom[v]) \}. 
 \end{align*}
 The coefficient of $Z_W^\wtarget Z_X^i$ in the corresponding polynomial $\Pin{v}{g} \in \FF_2[Z_W, Z_X]$ is given by
 \begin{align*}
  |\{(X, C) \in \partialin{v}{g} \sep \wfct(X) = \wtarget \text{ and } |X| = i\}| \mod 2.
 \end{align*}
 
 The coefficients of the monomials $Z_W^\wtarget Z_X^k$ for $0 \leq \wtarget \leq 2n^2$ in the polynomial $\Pin{r}{\emptyset}$ at the root node $r$ will then yield the desired numbers. We compute this polynomial by applying the equations \cref{eq:st_pex_leaf} to \cref{eq:st_pin}.
 In the base case, when $v$ is a leaf node in $\TT$, we can compute $\Pex{v}{f}$ by
 \begin{equation}
  \label{eq:st_pex_leaf}
  \begin{array}{lcl}
   \Pex{v}{f} & = & [(f^{-1}(\one_L), f^{-1}(\one_R)) \text{ is a consistent cut of } G[f^{-1}(\{\one_L, \one_R\})]] \\
              & \cdot & [T \cap \tail[v] \subseteq f^{-1}(\{\one_L, \one_R\})] [t_1 \in \tail[v] \rightarrow f(t_1) = \one_L]
  \end{array}
 \end{equation}
 which verifies that $f$ induces a valid partial solution. If $v$ is not a leaf node, then we compute $\Pex{v}{f}$ by the usual recurrence
 \begin{equation}
  \label{eq:st_pex_branch}
  \Pex{v}{f} = \prod_{u \in \child(v)} \Pin{u}{f}.
 \end{equation}
 To compute $\Pin{v}{g}$ we employ the recurrence
 \begin{equation}
  \label{eq:st_pin}
  \Pin{v}{g} = \Pex{v}{g[v \mapsto \zero]} + \Pex{v}{g[v \mapsto \one_L]}\, Z_W^{\wfct(v)} Z_X + \Pex{v}{g[v \mapsto \one_R]}\, Z_W^{\wfct(v)} Z_X.
 \end{equation}
 
 We proceed by proving the correctness of the equations \cref{eq:st_pex_leaf} to \cref{eq:st_pin}. Analogous to \textsc{Connected Vertex Cover}, whenever $f(t_1) \neq \one_L$ or $g(t_1) \neq \one_L$, we have for appropriate vertices $v$ that $\Pex{v}{f} = 0$ or $\Pin{v}{g} = 0$ respectively. Similarly, the corresponding polynomial will be $0$ whenever $f(t) = \zero$ or $g(t) = \zero$ for some terminal $t \in T$ due to Equation \cref{eq:st_pex_leaf}. Equation \cref{eq:st_pex_leaf} is correct, because when $v$ is a leaf in $\TT$ we have that $\tree(v) = \emptyset$ and hence $\partialex{v}{f}$ contains at most $(\emptyset, (\emptyset, \emptyset))$ whose membership is decided by the predicate in equation \cref{eq:st_pex_leaf}.
 
 The proof of correctness for equation \cref{eq:st_pex_branch} is straightforward and hence omitted. For the correctness of equation \cref{eq:st_pin} we consider a partial solution $(X, (X_L, X_R)) \in \partialin{v}{g}$ and distinguish three cases.
 \begin{enumerate}
  \item If $v \notin X$, then we have that $(X, (X_L, X_R)) \in \partialex{v}{f}$, where $f = g[v \mapsto \zero]$. Vice versa, any $(X, (X_L, X_R)) \in \partialex{v}{f}$ must also be in $\partialin{v}{g}$. Since $X$ does not change, we do not need to update the size or weight, hence we do not multiply by further formal variables in this case.
  
  \item If $v \in X_L \subseteq X$, then we claim that $(X \setminus \{v\}, (X_L \setminus \{v\}, X_R)) \in \partialex{v}{f}$, where $f = g[v \mapsto \one_L]$, because the definition of $X'$ and $C'$ in $\Pex{v}{f}$ remains unchanged. Vice versa, any partial solution in $\partialex{v}{f}$ can be extended to such a partial solution in $\partialin{v}{g}$ by adding $v$ to $X_L$. Since $|X| - |X \setminus \{v\}| = 1$ and $\wfct(X) - \wfct(X \setminus \{v\}) = \wfct(v)$, multiplication by $Z_W^{\wfct(v)} Z_X$ updates the size and weight correctly.
  
  \item If $v \in X_R \subseteq X$, the proof is analogous to case 2.
 \end{enumerate}
 Note that when $v = v_1$ equation \cref{eq:cvc_pin} simplifies and only the second case occurs. The running time and space bound follows from the general discussion in \cref{sec:gen_time_and_space}.
\end{proof}

\begin{thm}
 There exists a Monte-Carlo algorithm that given an elimination forest of depth $d$ solves \textsc{Steiner Tree} in time $\Oh^*(3^d)$ and polynomial space. The algorithm cannot give false positives and may give false negatives with probability at most 1/2.
\end{thm}

\begin{proof}
 We choose an arbitrary terminal $t_1 \in T$ and run \cref{algo:cutncount} with $U = V$ and the procedure $\countc$ as given by \cref{thm:st_count}. The correctness follows from \cref{thm:general_setup} and \cref{thm:st_correct}. The running time and space bound follow from \cref{thm:st_count}.
\end{proof}

Whereas for all other considered problems vertex costs are the more natural variant, for \textsc{Steiner Tree} we consider edge costs to be more natural. If we allow only polynomially-sized integral edge costs, we can reduce to the unweighted variant by replacing each edge $e$ by a path of length $c(e)$, where $c(e)$ is the cost of edge $e$.  This can only increase the treedepth by an additive term of $\Oh(\log(\max_{e \in E} c(e)))$, because, like in the proof of \cref{thm:fvs_fvs}, we can copy the previous elimination forest and attach at the leaves optimal elimination forests for the paths and we know that a path of length $\ell$ has treedepth at most $\Oh(\log \ell)$.

One can also consider \textsc{Connected Vertex Cover} as a special case of \textsc{Steiner Tree} by observing that the reduction presented by Cygan et al.~\cite{CyganNPPRW11} only increases the treedepth by 1.

\begin{lem}
 There is a polynomial-time algorithm that given a \textsc{Connected Vertex Cover} instance $(G,k)$ with an elimination forest $\TT$ of depth $d$ for $G$ constructs an equivalent \textsc{Steiner Tree} instance $(G', T, k')$ with an elimination forest $\TT'$ of depth $d + 1$ for $G'$.
\end{lem}

\begin{proof}
 To construct the instance $(G', T, k')$ we subdivide each edge of $G$ with a terminal, formally we replace $e = \{u, v\} \in E$ by a new vertex $w_e$ and two edges $\{u, w_e\}, \{w_e, v\}$ and take $T = \{w_e \sep e \in E(G)\}$. Setting $k' = |T| + k$, it is easy to see that $G$ contains a connected vertex cover of size $k$ if and only if $(G', T)$ contains a Steiner tree of cardinality at most $k' = |T| + k = E(G) + k$. 
 
 To obtain the elimination forest $\TT'$ of depth $d + 1$ for $G'$ first observe that $G' - T$ is an independent set. Hence, $\TT$ is also a valid elimination forest for $G' - T$. For an edge $e = \{u, v\} \in E(G)$, we can assume without loss of generality that $u$ is a descendant of $v$ in $\TT$ and attach a node corresponding to $w_e$ directly below $u$. Then, the endpoints of $\{u, w_e\}$ and $\{w_e, v\}$ are in an ancestor-descendant-relationship respectively. We repeat this process for all $w_e$ to obtain an elimination forest $\TT'$ of depth $d + 1$ for $G'$.
\end{proof}

\subsection{Connected Odd Cycle Transversal}

\begin{problem}[doublelined]{Connected Odd Cycle Transversal}
 Input: & An undirected graph $G = (V, E)$ and an integer $k$. \\
 Question: & Is there a set $X \subseteq V$, $|X| = k$, such that $G[X]$ is connected and $G - X$ is bipartite, i.e.\ there are $A$ and $B$ so that $V \setminus X = A \cupdot B$ and $E(G[A]) = E(G[B]) = \emptyset$?
\end{problem}

The algorithm for \textsc{Connected Odd Cycle Transversal} is mostly straightforward, but we need to be careful when defining the solutions and candidates. We say that $(A, B)$ is a \emph{bipartition} of $G = (V, E)$, denoted by $(A, B) \in \bip(G)$, if $A \cupdot B = V$ and $E(G[A]) = E(G[B]) = \emptyset$. For a set $X \subseteq V$ the counting procedure will test all ways of assigning the vertices in $G - X$ to $A$ and $B$ so that $(A, B)$ is a bipartition of $G - X$. This always results in an even number of possibilities and hence we cannot consider solely $X$ as a solution, because then also connected solutions would be counted an even number of times and thus cancel. To avoid this issue, we consider pairs $(X, A)$ as solutions and candidates, where $X$ is an odd cycle transversal and $A$ one side of a bipartition of $G - X$. Formally, we make the following definitions:
\begin{align*}
 \candidates & = \{ (X, A) \in 2^V \times 2^V \sep X \cap A = \emptyset, (A, V \setminus (X \cup A)) \in \bip(G - X), |X| = k \}, \\
 \solutions & = \{ (X, A) \in \candidates \sep G[X] \text{ is connected} \}.
\end{align*}

As in the case of \textsc{Connected Dominating Set}, we choose a vertex $v_1$ and fix it to the left side of the consistent cut and the algorithm will try all choices for $v_1$. Hence, we define the candidate-cut-pairs by
\begin{equation*}
 \cutsols = \{ (X,(X_L, X_R)) \in \cuts(V) \sep X \in \candidates \text{ and } v_1 \in X_L \}.
\end{equation*}

To make the Isolation Lemma (\cref{thm:isolation}) work in this case, we choose, similarly to \textsc{Feedback Vertex Set}, the universe $U = V \times \{ \deleted, \remain \}$. To interpret our candidates as subsets of $U$ we identify a pair $(X, A)$ with $(X \times \{\deleted\}) \cup (A \times \{\remain\})$ and a weight function $\wfct \colon U \rightarrow [N]$ associates to a candidate $(X, A)$ the weight $\wfct((X \times \{ \deleted \}) \cup (A \times \{ \remain \}))$. 

\begin{thm}[\cite{CyganNPPRW11}]
 \label{thm:coct_correct}
 Let $\wfct \colon U \rightarrow [N]$ be a weight function, and $\cutsols$ and $\solutions$ as defined above. Then we have for every $\wtarget \in \NN$ that $|\solutions_\wtarget| \equiv |\cutsols_\wtarget|$, where $\solutions_\wtarget = \{(X, A) \in \solutions \sep \wfct((X \times \{\deleted\}) \cup (A \times \{\remain\})) = \wtarget\}$ and $\cutsols_\wtarget = \{((X,A),C) \in \cutsols \sep \wfct((X \times \{\deleted\}) \cup (A \times \{\remain\})) = \wtarget\}$.
\end{thm}

\begin{proof}
 \cref{thm:cons_cut} implies that $|\cutsols_\wtarget| = \sum_{(X, A) \in \candidates_\wtarget} 2^{\cc(G[X])-1}$, where $\candidates_\wtarget = \{(X,A) \in \candidates \sep \wfct((X \times \{\deleted\}) \cup (A \times \{\remain\})) = \wtarget \}$. Hence, $|\cutsols_\wtarget| \equiv |\{(X,A) \in \candidates_\wtarget \sep \cc(G[X]) = 1\}| = |\solutions_\wtarget|$.
\end{proof}

We can now present the counting procedure for \textsc{Connected Odd Cycle Transversal}.

\begin{lem}
 \label{thm:coct_count}
 Given a connected graph $G = (V, E)$, terminals $T \subseteq V$, an integer $k$, a weight function $\wfct \colon U \rightarrow [4n]$, and an elimination forest $\TT$ of depth $d$ for $G$, we can determine $|\cutsols_\wtarget|$ modulo 2 for all $0 \leq \wtarget \leq 4n^2$ in time $\Oh^*(4^d)$ and polynomial space.
\end{lem}

\begin{proof}
 For induced subgraphs $G[V']$, where $V' \subseteq V$, we define the set of candidates and candidate-cut-pairs by:
 \begin{equation*}
 \begin{array}{lll}
  \widehat{\candidates}(V') & = \{ (X, A) \in 2^{V'} \times 2^{V'} \sep & X \cap A = \emptyset, (A, V' \setminus (X \cup A)) \in \bip(G[V'] - X)\}, \\
  \widehat{\cutsols}(V') & = \{ ((X, A), (X_L, X_R)) \sep & (X, A) \in \widehat{\candidates}(V'),  (X, (X_L, X_R)) \in \cuts(V'), \\
 & & \text{ and } (v_1 \in V' \rightarrow v_1 \in X_L)\}.
 \end{array}
 \end{equation*}

 The states are given by $\states = \{\zero_A, \zero_B, \one_L, \one_R\}$, where $\zero_A$ and $\zero_B$ represent vertices that are not in the odd cycle transversal and the subscript denotes to which side of the bipartition they belong to. The states $\one_L$ and $\one_R$ represent vertices inside the odd cycle transversal on the respective side of the consistent cut. For every vertex $v$ and assignment $f \colon \tail[v] \rightarrow \{\zero_A, \zero_B, \one_L, \one_R\}$ we define the partial solutions at $v$, but excluding $v$, that respect $f$ by 
 \begin{align*}
  \partialex{v}{f} = \{ ((X, A), (X_L, X_R)) \in \widehat{\cutsols}(\tree(v)) \sep\, & X' = X \cup f^{-1}(\{\one_L, \one_R\}), \\
                                                                & C' = (X_L \cup f^{-1}(\one_L), X_R \cup f^{-1}(\one_R)), \\
                                                                & A' = A \cup f^{-1}(\zero_A), \\
                                                                & ((X', A'), C') \in \widehat{\cutsols}(\broom[v]) \}.
 \end{align*}
 The coefficient of $Z_W^\wtarget Z_X^i$ in the associated polynomial $\Pex{v}{f} \in \FF_2[Z_W, Z_X]$ is given by
 \begin{equation*}
  |\{((X,A),(X_L, X_R)) \in \partialex{v}{f} \sep \wfct((X \times \{\deleted\}) \cup (A \times \{\remain\})) = \wtarget, |X| = i \}| \mod 2.
 \end{equation*}
 
 Similarly, for every vertex $v$ and assignment $g \colon \tail(v) \rightarrow \{\zero_A, \zero_B, \one_L, \one_R\}$ we define the partial solutions at $v$, possibly including $v$, that respect $g$ by
 \begin{align*}
  \partialin{v}{g} = \{ ((X, A), (X_L, X_R)) \in \widehat{\cutsols}(\tree[v]) \sep\, & X' = X \cup g^{-1}(\{\one_L, \one_R\}), \\
                                                                & C' = (X_L \cup g^{-1}(\one_L), X_R \cup g^{-1}(\one_R)), \\
                                                                & A' = A \cup g^{-1}(\zero_A), \\
                                                                & ((X', A'), C') \in \widehat{\cutsols}(\broom[v]) \}.
 \end{align*}
 The coefficient of $Z_W^\wtarget Z_X^i$ in the associated polynomial $\Pin{v}{g} \in \FF_2[Z_W, Z_X]$ is given by
 \begin{equation*}
  |\{((X,A),(X_L, X_R)) \in \partialin{v}{g} \sep \wfct((X \times \{\deleted\}) \cup (A \times \{\remain\})) = \wtarget, |X| = i \}| \mod 2.
 \end{equation*}
 The exponents of the monomials $Z_W^\wtarget Z_X^i$ range between $0 \leq \wtarget \leq 4n^2$ and $0 \leq i \leq n$. Extracting the coefficients of the monomials $Z_W^\wtarget Z_X^k$ for $0 \leq \wtarget \leq 4n^2$ in the polynomial $\Pin{r}{\emptyset}$ at the root node $r$ yields the desired numbers. The polynomials $\Pex{v}{f}$ and $\Pin{v}{g}$ are computed by applying the three upcoming equations. To compute $\Pex{v}{f}$ in the case that $v$ is a leaf node in $\TT$, we only need to verify that $f$ induces a valid partial solution, yielding the equation 
 \begin{equation}
  \label{eq:coct_pex_leaf}
  \begin{array}{lcl}
  \Pex{v}{f} & = & [(f^{-1}(\one_L), f^{-1}(\one_R)) \text{ is a consistent cut of } G[f^{-1}(\{\one_L, \one_R\})]] \\     
             & \cdot & [(f^{-1}(\zero_A), f^{-1}(\zero_B)) \text{ is a bipartition of } G[f^{-1}(\{\zero_A, \zero_B\})]] \\
             & \cdot & [v_1 \in \tail[v] \rightarrow f(v_1) = \one_L].
  \end{array}
 \end{equation}
 When $v$ is not a leaf, we can compute $\Pex{v}{f}$ by the usual recurrence
 \begin{equation}
  \label{eq:coct_pex_branch}
  \Pex{v}{f} = \prod_{u \in \child(v)} \Pin{u}{f}.
 \end{equation}
 The polynomial $\Pin{v}{g}$ is computed by the recurrence
 \begin{equation}
  \label{eq:coct_pin}
  \begin{array}{rlcl}
   \Pin{v}{g} = & \Pex{v}{g[v \mapsto \zero_A]} Z_W^{\wfct(v, \remain)} & + & \Pex{v}{g[v \mapsto \zero_B]}\\
              + & \Pex{v}{g[v \mapsto \one_L]} Z_W^{\wfct(v, \deleted)} Z_X & + & \Pex{v}{g[v \mapsto \one_R]} Z_W^{\wfct(v, \deleted)} Z_X,
  \end{array}
 \end{equation}
 which tests all four states of $v$ in a candidate-cut-pair.

 We proceed by proving the correctness of equations \cref{eq:coct_pex_leaf} to \cref{eq:coct_pin}. Similar to \textsc{Connected Vertex Cover}, if $f(v_1) \neq \one_L$ or $g(v_1) \neq \one_L$ for appropriate $f$ or $g$, then $\Pex{v}{f} = 0$ or $\Pin{v}{g} = 0$ respectively. To see the correctness of equation \cref{eq:coct_pex_leaf}, note that $\partialex{v}{f}$ can contain at most $((\emptyset, \emptyset), (\emptyset, \emptyset))$ when $v$ is a leaf and verifying its containment reduces to the three predicates on the right-hand side of equation \cref{eq:coct_pex_leaf}. 
 
 The proof of correctness for equation \cref{eq:coct_pex_branch} proceeds as usual. Here, when combining disjoint partial solutions we take the union of the associated bipartitions and hence crucially rely on the property that there are no edges between $\tree[u]$ and $\tree[u']$ for $u \neq u' \in \child(v)$, as otherwise the result need not be a bipartition again. We omit further details.
 
 It remains to prove the correctness of equation \cref{eq:coct_pin}. We consider a partial solution $((X, A), (X_L, X_R)) \in \partialin{v}{g}$ and distinguish between four cases based on the state of $v$ in this partial solution.
 \begin{enumerate}
  \item If $v \in A$, then $((X, A \setminus \{v\}), (X_L, X_R)) \in \partialex{v}{f}$, where $f = g[v \mapsto \zero_A]$, because this preserves the definition of $X'$, $C'$, and $A'$ in the predicate of $\partialin{v}{g}$ and $\partialex{v}{f}$. Vice versa, any partial solution in $\partialex{v}{f}$ can be extended to a partial solution in $\partialin{v}{g}$ by adding $v$ to $A$. This increases the weight of $(X,A)$ by $\wfct(v, \remain)$ and thus we must multiply by $Z_W^{\wfct(v, \remain)}$ to keep track of this change.
  \item If $v \notin X \cup A$, then $((X,A),(X_L, X_R)) \in \partialin{v}{g}$ if and only if $((X,A),(X_L, X_R)) \in \partialex{v}{f}$, where $f = g[v \mapsto \zero_B]$. Since $(X,A)$ is not changed, we do not need to multiply by formal variables in this case.
  \item If $v \in X_L \subseteq X$, then $((X \setminus \{v\}, A), (X_L \setminus \{v\}, X_R)) \in \partialex{v}{f}$, where $f = g[v \mapsto \one_L]$, because as in case 1 and 2 this preserves the definition of $X'$, $C'$, and $A'$. Vice versa, any partial solution in $\partialex{v}{f}$ can be extended to a partial solution in $\partialin{v}{g}$ by adding $v$ to $X_L$. This increases the weight of $(X, A)$ by $\wfct(v,\deleted)$ and the size of $X$ by 1, hence this is tracked correctly by multiplication with $Z_W^{\wfct(v, \deleted)} Z_X$.
  \item If $v \in X_R \subseteq X$, then the proof is analogous to case 3.
 \end{enumerate}
 If $v = v_1$, then only the third case occurs and, by an earlier discussion, equation \cref{eq:coct_pin} simplifies to $\Pin{v}{g} = \Pex{v}{g[v \mapsto \one_L]} Z_W^{\wfct(v, \deleted)} Z_X$. The running time and space bound follows from the general discussion in \cref{sec:gen_time_and_space}. 
\end{proof}

\begin{thm}
 There exists a Monte-Carlo algorithm that given an elimination forest of depth $d$ solves \textsc{Connected Odd Cycle Transversal} in time $\Oh^*(4^d)$ and polynomial space. The algorithm cannot give false positives and may give false negatives with probability at most 1/2.
\end{thm}

\begin{proof}
 We let $v_1$ iterate through all vertices. For each choice of $v_1$, we run \cref{algo:cutncount} with $U = V \times \{\deleted, \remain\}$ and the procedure $\countc$ as given by \cref{thm:coct_count}. The correctness follows from \cref{thm:general_setup} and \cref{thm:coct_correct}. The running time and space bound follows from \cref{thm:coct_count}.
\end{proof}

\section{Conclusion}
\label{sec:conclusion}

The Cut\&Count technique of Cygan et al.~\cite{CyganNPPRW11} has provided single-exponential-time and -space algorithms for many connectivity problems parameterized by treewidth. We have shown that this technique is just as useful for parameterization by treedepth, where we have obtained single-exponential-time and \emph{polynomial-space} algorithms. Our algorithms run in time $\Oh^*(\alpha^{d})$, where $\alpha$ is a small constant and $d$ is the depth of a given elimination forest. The base $\alpha$ matches that obtained by Cygan et al.~\cite{CyganNPPRW11} for parameterization by treewidth. Assuming SETH, this base is optimal for treewidth, or even pathwidth \cite{CyganNPPRW11}. In principle, since treedepth is a larger parameter than both treewidth and pathwidth, it may be possible to obtain better running times when parameterizing by treedepth, possibly at the cost of using exponential space. The style of construction, used to obtain lower bounds relative to treewidth, used by Lokshtanov et al.~\cite{LokshtanovMS18} and Cygan et al.~\cite{CyganNPPRW11}, necessitates long paths and is thereby unsuitable for bounds relative to treedepth. Thus, the question remains whether our running times are optimal; it is tempting to conjecture that they are.

While we have not given the proofs, our techniques also extend to other problems like \textsc{Connected Feedback Vertex Set} and \textsc{Connected Total Dominating Set}. However, there are several problems, including \textsc{Cycle Cover} and \textsc{Longest Cycle}, for which Cygan et al.~\cite{CyganNPPRW11} obtain efficient algorithms, where it is yet unclear how to solve them in polynomial space when parameterizing by treedepth. In particular, \textsc{Hamiltonian Path} and \textsc{Hamiltonian Cycle} share the same issues, namely that the algorithms parameterized by treewidth keep track of the degrees in the partial solutions and it is not clear how to do that when branching on the elimination forest while only using polynomial space. Belbasi and F\"urer~\cite{BelbasiF19} can count Hamiltonian cycles in polynomial space, but their running time also depends on the width of a given tree decomposition. An algorithm for any of these problems parameterized by treedepth, with single-exponential running time and requiring only polynomial space, would be quite interesting. As mentioned in \cref{followup} this has by now been solved by Nederlof et al.~\cite{NederlofPSW20}.

\newpage

\bibliography{cutandcount}

\end{document}